\documentclass[journal]{IEEEtran}
\ifCLASSINFOpdf
\else
\fi

\usepackage[hidelinks]{hyperref}   
\usepackage{hyperref}
\usepackage{subcaption}
\usepackage{epstopdf}
\usepackage{graphicx}
\usepackage{amsfonts}
\usepackage{enumerate}
\usepackage{varioref}
\usepackage{longtable}
\usepackage{graphicx}
\usepackage{amssymb}
\usepackage{chngpage}
\usepackage{multirow}
\usepackage{multirow}
\usepackage{amsthm}
\usepackage{amsmath}
\usepackage{caption}
\usepackage{nccmath}  
\usepackage[table]{xcolor}  
\usepackage{upgreek}
\usepackage{tipa}
\usepackage{algorithm}
\usepackage{algpseudocode}
\usepackage{threeparttable}
\definecolor{bananamania}{rgb}{0.98, 0.91, 0.71}
\usepackage{mathtools}  
\DeclarePairedDelimiter{\ceil}{\lceil}{\rceil}  

\usepackage[utf8]{inputenc}
\usepackage[english]{babel}
\usepackage{amsthm} 
\usepackage{color}
\usepackage{todonotes}
\usepackage{verbatim}
\usepackage{algorithm}
\usepackage{algpseudocode}
\usepackage{amssymb}  
\usepackage{amsmath}   
\usepackage{graphicx}  
\usepackage{subcaption}  

\newtheorem{definition}{Definition}
\newtheorem{example}{Example}
\theoremstyle{assumption}

\newtheorem{proposition}{Proposition}
\newtheorem{corollary}{Corollary}
\usepackage{color}

\usepackage[margin=0.8in]{geometry}

\usepackage{eso-pic} 

\begin{document}

\AddToShipoutPictureBG*{%
  \AtPageUpperLeft{%
    \setlength\unitlength{1in}%
    \hspace*{\dimexpr0.5\paperwidth\relax}
    \makebox(0,-0.75)[c]{\normalsize Published in Signal Processing journal, Elsevier, 2020.}
    }}

\title{Distributed Voting in Beep Model}

\author{Benyamin Ghojogh and
		Saber Salehkaleybar

\thanks{Benyamin Ghojogh's e-mail:  \href{mailto:ghojogh_benyamin@ee.sharif.edu}{ghojogh\_benyamin@ee.sharif.edu}}         
\thanks{Benyamin Ghojogh's main e-mail: bghojogh@uwaterloo.ca}   
\thanks{Saber Salehkaleybar's e-mail: saleh@sharif.edu}
\thanks{Both authors are with Department of Electrical Engineering, Sharif University of Technology, Tehran, Iran.}}


\maketitle

\begin{abstract}
We consider the problem of distributed multi-choice voting in a setting that each node can communicate with its neighbors merely by sending beep signals. Given its simplicity, the beep communication model is of practical importance in different applications such as system biology and wireless sensor networks. Yet, the distributed majority voting has not been resolved in this setting. In this paper, we propose two algorithms, named Distributed Voting with Beeps, to resolve this problem. In the first proposed algorithm, the adjacent nodes having the same value form a set called spot. Afterwards, the spots with majority value try to corrode the spots with non-majority values. The second proposed algorithm is based on pairwise interactions between nodes. The proposed algorithms have a termination detection procedure to check whether voting is achieved. We establish theoretical guarantees for the convergence of these algorithms. In particular, we show that the success probability of the first algorithm tends to one as the percentage of the initial votes in majority increases. For the second algorithm, we show that it returns the correct output with high probability. Our experiments show that the algorithms are fairly invariant to the network topology, initial distribution of values, and network size.
\end{abstract}


\IEEEpeerreviewmaketitle

\section*{Highlights}

\begin{itemize}
\item Proposing two distributed algorithms for multi-choice majority voting in the beep model
\item Showing the correctness of the algorithms and analyzing their time,
message, and space complexities
\item Comparing with the baseline methods analytically
\item Evaluating the proposed methods experimentally on different network topologies and initial distributions of votes
\end{itemize}

\section{Introduction}\label{section_introduction}

\IEEEPARstart{C}{ommunication}  protocols fall into two main categories, i.e., shared memory and message-passing  \cite{lynch1996distributed,navlakha2015distributed}. In the shared memory model, agents communicate with each other by writing in or reading from a global shared memory whereas message-passing communication is in the form of sending and receiving messages between agents.  
The message-based communication protocols can be categorized into different types such as population protocol, stone-age protocol, and beep model \cite{navlakha2015distributed}. 

In population protocols, pairs of agents, modeled as finite state machines, communicate with each other to update their states. The order of asynchronous pairwise interactions between agents is determined by a fair scheduler \cite{aspnes2009introduction}. 
In the stone-age protocol, however, the interactions are not necessarily pairwise. This protocol is based on a network of finite state machines (nFSM) \cite{emek2013stone} which allows the agents to communicate asynchronously through rich-sized messages while restricting the capacity of their memory to a constant value \cite{navlakha2015distributed}. In contrast, in the beep model \cite{cornejo2010deploying},  each agent has an internal memory of logarithmic size while the messages are binary \cite{navlakha2015distributed}. This model is synchronous and each node can merely beep or listen at every time slot \cite{cornejo2010deploying,navlakha2015distributed}. 
Examples of other distributed computation models with limited communication/memory capabilities include  LOCAL \cite{linial1987distributive} and CONGEST \cite{peleg2000distributed}, in which agents interact at synchronous time slots.
In the LOCAL model, compared to the beep model, there is no restriction on the size of messages and each node has a unique ID. The CONGEST model is similar to the LOCAL model except that the size of message is restricted to $\Uptheta(\log(n))$ where $n$ is the number of agents.



Distributed systems with severe communication constraints have been studied in system biology \cite{navlakha2015distributed,feinerman2013theoretical} and wireless sensor networks \cite{cornejo2010deploying}. 
Some of the examples in system biology are in Chemical Reaction Networks (CRN) \cite{chen2014deterministic,chen2013programmable}, protein-protein interactions \cite{uetz2000comprehensive}, and natural animal interactions \cite{beauquier2013tight,guerraoui2015byzantine}. 
The population protocol, modeling asynchronous pairwise communication between mobile agents, can be commonly seen in the nature \cite{navlakha2015distributed}; a real world example involving natural life is given in \cite{beauquier2013tight}. 
This is in part because in biological systems, the memory of a biological cell is highly limited \cite{cheong2011information}. Hence, the stone-age protocol can be a good model for biological communications \cite{navlakha2015distributed}.
Moreover, biological interactions carry small effective information \cite{jongeneel2005atlas}, which motivates the  use of  beep model for biological communication \cite{navlakha2015distributed}. 
Previous works have studied the application of this model for distributed voting among cells and DNA molecules \cite{chen2013programmable}.




\subsection{Related Work}
Although beep model is almost a newly developed communication protocol, it has received much attention in recent years. Different algorithms have been proposed for various distributed problems in this model. Examples of these problems include interval coloring and graph coloring \cite{cornejo2010deploying,halldorsson2014distributed,casteigts2019design,metivier2015distributed}, leader election \cite{dufoulon2018brief,dufoulon2018beeping,dufoulon2018beepingDISC,forster2014deterministic,ghaffari2013near,czumaj2016brief,augustine2013robust,gilbert2015computational}, maximal independent set \cite{halldorsson2018computing,afek2013beeping,holzer2017beeping,halldorsson2015distributed,scott2013feedback,jeavons2016feedback}, minimum connected dominating set \cite{yu2015minimum}, network size approximation and counting \cite{brandes2016approximating,brandes2017fast,casteigts2019counting,kardas2014distributed}, deterministic rendezvous problem \cite{elouasbi2017deterministic}, naming problem \cite{chlebus2017naming}, membership problem \cite{huang2013conflict,huang2012brief}, broadcasting \cite{hounkanli2016asynchronous,czumaj2016communicating,czumaj2019communicating,hounkanli2015deterministic}, and consensus \cite{hounkanli2016global}.

In the problem of consensus, each agent chooses an initial value and and the goal is that all agents end up having a common value from a set of initial values. In \cite{hounkanli2016global}, a distributed consensus algorithm was proposed for the beep model in a fully-connected network. 
Authors in \cite{hounkanli2016global} first proposed an algorithm for obtaining global synchronization out of local clocks of agents. 
For achieving consensus, they transformed the binary representation of values to a bit stream. After achieving synchronization, every agent beeps and listens in several defined time slots according to its bit stream.
If an agent does not hear any beep, it infers that all the agents have the same value so it outputs its own value; otherwise, it outputs the default value. 
The agents can reach consensus in time $O(\log l_{\text{min}})$, where $l_{\text{min}}$ is the smallest initial value among agents \cite{hounkanli2016global}. 

Distributed voting is a special type of the consensus problem in which all the agents should decide on a value selected by majority of agents as their initial values. 
Distributed voting has various applications in distributed signal processing. Some possible applications are probabilistic belief propagation in image processing applications (e.g., where the image can be considered as a factor graph) \cite{sudderth2008signal,komodakis2007image}, signal processing over the peer-to-peer networks \cite{wang2007distributed,dong2019pair}, and consensus in distributed pattern analysis and statistical learning \cite{jordan2019communication,blot2019distributed}. These applications are mainly introduced in wireless sensor networks, where sensors collect data locally and distributed voting is required as a subroutine.
As another example, wireless networks with severely limited communication capabilities can be modelled by the beep model. For instance, Cornejo and Kuhn \cite{cornejo2010deploying} utilized  a ``carrier sensing" mechanism in wireless networks in order to send beep signals in a sensor network.

In this paper, we consider the problem of distributed voting algorithms in the beep model. 
There are several distributed voting algorithms in the literature for other types of communication protocols. 
Some of these algorithms have been proposed for the binary case in population protocols. Binary voting using randomized gossip algorithm \cite{boyd2006randomized}, average consensus \cite{benezit2009interval}, two-sample voting \cite{cooper2014power}, and via exponential distribution \cite{salehkaleybar2016distributed} are the examples of binary voting algorithms. 
There are also some methods for multi-choice voting, such as voting using pairwise asynchronous graph automata \cite{benezit2011distributed} and union/intersection operations \cite{salehkaleybar2015distributed}. 
Moreover, there are voting algorithms which perform voting through ranking \cite{salehkaleybar2015distributed,jung2012distributed} and plurality consensus \cite{becchetti2017simple,becchetti2015plurality,babaee2013distributed}.


\subsection{Our Contribution}

In this paper, we propose Distributed Voting with Beeps 1 (DVB1) and Distributed Voting with Beeps 2 (DVB2) algorithms for \textbf{multi-choice voting} in the \textbf{beep model}  which work under any \textbf{arbitrary network topology}. To the best of our knowledge, these are the first distributed algorithms for the problem of majority voting in this model. The two DVB algorithms have \textbf{simple} structures and are applicable in both wireless networks and biological networks with limited message size. The proposed algorithms have the following characteristics:
\begin{itemize}
\item 
Let $n_1$ and $n_2$ be the number of supporters of the first and second most  popular initial values. In a fully connected network with $K$ possible choices, DVB1 algorithm returns the correct result with probability at least $(1-\exp(-\sqrt{n_1/n_2})\times\exp(-(K-1)n_2/(\sqrt{n_1 n_2}-1))$.     
Thus, the probability of success in DVB1 goes to one as the ratio $n_1/n_2$ increases. The DVB2 also returns the correct vote with high probability.
\item Our empirical results show that the DVB1 algorithm is  fairly invariant to the network topology, initial distribution of values, and the population of nodes. DVB2 is also invariant to initial values and population according to \cite{salehkaleybar2015distributed}.
\item Both algorithms incorporate termination detection for voting as well as the distributed consensus. 
\item The time, message, and space complexities of DVB1 algorithm are $\Uptheta(KD\log(N))$, $O(ND\log(N) + NK)$, and $O(\log(KD) + \log^{(2)}(N))$, respectively, where $K$ is the number of possible choices, $D$ is the diameter of the network, $N$ is the number of nodes, and $\log^{(2)} (.) = \log (\log (.))$. 
The lower bound on the message and space complexities of DVB1 are $\Omega(ND)$ and $\Omega(\log(KD))$, respectively. 
The time, message, and space complexities of DVB2 algorithm are $O\big(D\Delta^2 (\log(\Delta))^2 \log(N) + KD\Delta \log(\Delta) \log(N) \big)$, $O(KDN\log(N))$, and $O(K\log(K) + \log(D\Delta))$, respectively, where $\Delta$ is the maximum degree in the network. The lower bound on the time, message, and space complexities of DVB2 are $\Omega(\log(N))$, $\Omega(ND + N \log(N))$, and $\Omega(\log(KD) + \log^{(2)} (\Delta))$, respectively.

\end{itemize}

The remainder of this paper is organized as follows. Section \ref{section_systemModel} describes the beep model and defines the distributed voting problem. The description of DVB1 and DVB2 algorithms are given in Sections \ref{section_DVB1} and \ref{section_DVB2}, respectively. 
The correctness of DVB algorithms, as well as the probability of success of DVB1 algorithm, are provided in Section \ref{section_correctness}.
The analyses of time, message, and space complexities of the two algorithms are given in Section \ref{section_complexity}. Comparison to alternative approaches is reported in Section \ref{section_comparison}. The proposed algorithms are evaluated experimentally in Section \ref{section_simulation}. Finally, Section \ref{section_conclusion} concludes the paper and discusses possible future directions. 

\section{System Model}\label{section_systemModel}

In this section, we describe the system model and introduce
the concepts that are used in the description of the proposed
algorithms.

\subsection{Network Model}

Consider a network of $N$ nodes where each node has an initial value. The network topology can be determined by a graph $G = (V,E)$ where $V$ is the set of vertices (nodes) as $V = \{1,2,...,N\}$ and $E$ denotes the set of edges existing in the network. The edge set can be determined as $E \subseteq V \times V$, such that $(i,j) \in E$ if and only if nodes $i$ and $j$ can communicate directly. We say that nodes $i$ and $j$ are neighbors if $(i,j)
\in E$. The edges of network are all bidirectional. 
The diameter of network is denoted by $D$. We assume that upper bounds on $N$ and $D$ (or just on $N$ which is also an upper bound on $D$) are already known by nodes.

\subsection{Beep Model}
We assume that nodes are communicating with each other based on the beep model \cite{cornejo2010deploying}. In this model, each node can send a beep signal and all its neighbor nodes will receive it. Note that every node can either beep or listen; therefore, listening while beeping is not possible in this model. Every node cannot distinguish the number of beeps it receives if these beep signals are sent at the same time.
In this model, the internal memory of nodes is restricted to a logarithmic amount of bits.
We assume that the global synchronization has been achieved, as also assumed in the original beep model \cite{cornejo2010deploying}, and the nodes send beep signals at the specified time slots; hence,  nodes do not need to keep the global time index in memory.
In this work, the basic unit of time is referred to as time slot.

\subsection{Problem Definition}

We denote the value of node $i$ at time slot $t$ by $v_i(t)$. Moreover, we consider a set of $K$ levels for the values of nodes, i.e., $L= \{l_1,l_2, ..., l_K\}$.
We assume that all nodes know the number of levels, i.e., $K$.

\begin{definition}
We denote the number of nodes having value $l_j$ at time slot $t$ by $\#l_j(t)$. Thus, we have: $\#l_j(t) \triangleq \big|\{i|i \in V, v_i(t) = l_j\}\big|$. 
\end{definition}

Our goal is to design distributed algorithms based on the beep model such that all nodes find a level $l_k$ which is chosen by the majority of nodes as their initial values. In other words, we want to determine the level $l_k$ where $\#l_k(0) \geq \#l_i(0)$, $\forall i \neq k$.

\begin{definition}
In this work, the time, message, and space complexities are defined as the number of time slots, the number of total beeps transimitted by the nodes, and the number of bits required by every node to run the algorithm until reaching consensus on the majority vote, respectively.
\end{definition}

\section{Distributed Voting with Beeps (DVB) 1 Algorithm}\label{section_DVB1}
The DVB1 algorithm is based on two concepts which we call ``spot'' and ``corrosion of spots''. First, we explain these concepts and then we describe DVB1 algorithm.

\begin{figure}[!t]
\centering
\includegraphics[width=\columnwidth]{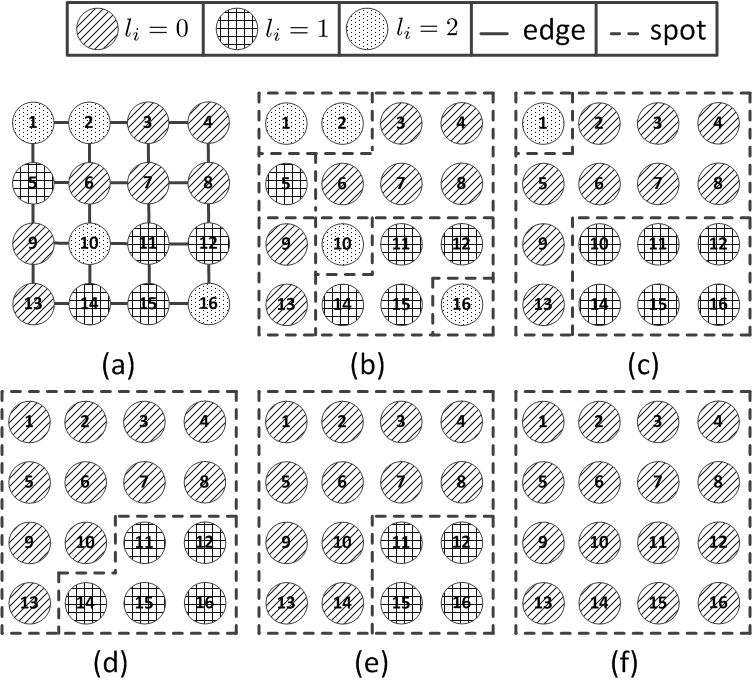}
\caption{An example of executing the DVB1 algorithm in 2D mesh grid topology. (a) initial values of nodes, (b) initial existing spots, (c) possible values after  the first corrosion, (d) possible values after the second corrosion, (e) possible values after the third corrosion, (f) eventual values after termination where sufficient number of phases have passed.}
\label{figure_example}
\end{figure}

\subsection{Spots}

\begin{definition}
A path between nodes $i$ and $j$ is denoted by $i \rightsquigarrow j: i \rightarrow i_1 \rightarrow i_2 \rightarrow \dots  \rightarrow i_r \rightarrow j$, $\{i_1, \dots, i_r\} \subseteq V$ where any two nodes $i_k$ and $i_{k+1}$ on the path, $1\leq k\leq r-1$, are neighbors. 
\end{definition}

\begin{definition}\label{definition_spot}
We say a set of nodes $V'\subseteq V$ form a spot if:
\begin{align*}
&\forall i,j \in V', \exists \text{ } i \rightsquigarrow j: i \rightarrow i_1 \rightarrow i_2 \rightarrow \dots  \rightarrow i_r \rightarrow j, \\&\mbox{such that }
 v_i(t) = v_{i_1}(t) = v_{i_2}(t) = \dots = v_{i_k}(t) = v_j(t).
\end{align*}
\end{definition}

Note that several spots may exist in a network. Furthermore, there might be a node such that its value is different from all its neighbors. In this case, that node is forming a spot only by itself.
It is noteworthy to mention that some similar concepts to ``spot'' can be found in the literature;  for example, ``cluster'' in \cite{ghaffari2013near} is defined as a set of nodes including one candidate node for leadership.

\begin{example}
As an example, suppose we have three different values $l_0$, $l_1$, and $l_2$ in the network. 
Figure \ref{figure_example}-a depicts an example of initial values in a 2D mesh grid network. As shown in Fig. \ref{figure_example}-b, the initial values form spots in the network. This network has spots of nodes $\{1,2\}$, $\{3,4,6,7,8\}$, $\{5\}$, $\{9,13\}$, $\{10\}$, $\{11,12,14,15\}$, $\{16\}$. As it can be seen, some of the spots, e.g. $\{5\}$, are single-node spots. 
\end{example}

\subsection{Corrosion Of Spots}

The main idea of the DVB1 algorithm is to shrink small spots, i.e., spots with a few number of nodes. We call this procedure ``corrosion". After multiple corrosions, large spots with the same value will merge and create one large spot containing all the nodes in the network. To draw an analogy, suppose that there are two neighboring spots and one of them is smaller than the other one. Imagine that the small spot is like a material which is sensitive to a special type of acid. The larger spot can be considered as the acid which corrodes the small spot until it disappears and the acid remains solely. 

We propose DVB1 algorithm to corrode small spots. The DVB1 algorithm consists of two main subroutines ``Corrosion()'' and ``TerminationDetection()'' (see Algorithm \ref{algorithm_overall} which is performed by every node $i$). The proposed algorithm executes iteratively until it terminates. We call each iteration of the algorithm a phase. In DVB1 algorithm, each phase consists of $T$ rounds for the corrosion subroutine. Every $D$ phases, the termination detection subroutine is executed having at most $(K-1)$ periods (see Fig. \ref{figure_diagram}). 
As shown in this figure, a round and a period consist of $K$ and $(D + 1)$ time slots, respectively.
The timing details of algorithm will be explained in the following.
Note that, in order to improve time complexity of algorithm, the termination detection is run every $D$ phases, i.e., at the end of every $O(D)$ phases.

\begin{algorithm}[!h]
\caption{DVB1 Algorithm}\label{algorithm_overall}
\begin{algorithmic}[1]
\State counter $\gets$ $0$
\Repeat  \Comment every iteration is a phase
    \State counter $\gets$ counter $+ 1$
	\State Corrosion()
	\If{counter $= D$}
	    \State counter $\gets$ $0$  \Comment reset the counter
	    \State TerminationDetection() \label{algorithmOverall_termination}
	\EndIf
\Until{not terminated}
\end{algorithmic}
\end{algorithm}


\begin{definition}
A round $R_j$ is a sequence of $K$ time slots $(t_{j,1}, t_{j,2}, \dots, t_{j,K})$ in which time slot $t_{j,k}$ is allocated to the level $l_k$. Note that, in this work, the term ``round'' should not be confused with the basic unit of time which is common in the literature (here, the basic unit of time is named ``time slot'').
\end{definition}

In round $R_j$, at time slot $t_{j,k}$, each node having the value of $l_k$ beeps if it is alive in that round. Then, it decides to either remain alive or die for the next round with probability $p$ ($=\frac{1}{2}$) (see lines \ref{algorithmCorrosion_start_beep}-\ref{algorithmCorrosion_end_beep} in Algorithm \ref{algorithm_corrosion}). 
In line \ref{algorithmCorrosion_probability_beep} in Algorithm \ref{algorithm_corrosion}, $U(0,1)$ denotes the uniform random number in the range $[0,1]$. 
Please note that a node cannot beep until the end of corrosion in a phase if it dies after beeping in a round of that phase. In Algorithm \ref{algorithm_corrosion}, we check this condition using the variable ``IsAllowedToBeep''.

\begin{figure}[!t]
\centering
\includegraphics[width=\columnwidth]{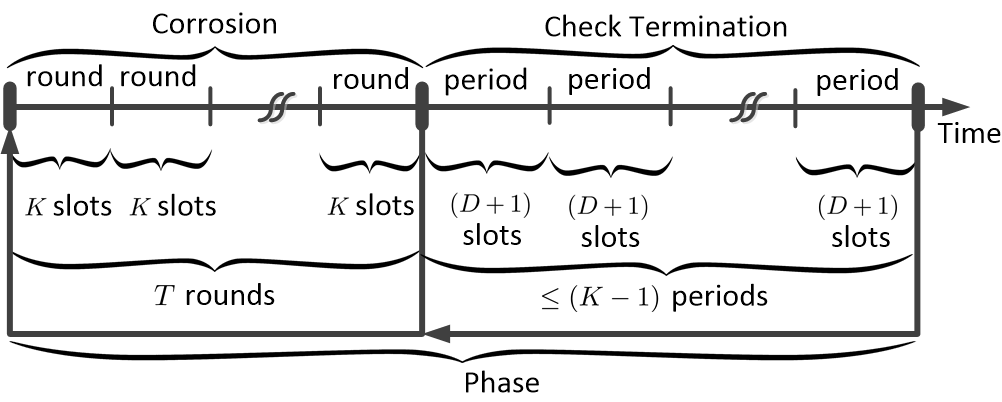}
\caption{A schematic of the phase in the DVB1 algorithm. The termination detection is run every $D$ phases, i.e., at the end of every $O(D)$ phases.}
\label{figure_diagram}
\end{figure}

\begin{algorithm}[!h]
\caption{Corrosion}\label{algorithm_corrosion}
\begin{algorithmic}[1]
	\State IsAllowedToBeep $\gets$ 1
	\For{$j$ \textbf{from} 1 \textbf{to} $T$}  \Comment{$T = \lceil c_1 \log_2 (N) \rceil$}
		\State hearBeep($l_k$) $\gets$ 0, $\forall k=1,\dots, K$ 
		\For{$k$ \textbf{from} $1$ \textbf{to} $K$} \label{algorithmCorrosion_timeSlot}
			\If{$v_i(t_{j,k})$ = $l_k$ \textbf{and} IsAllowedToBeep = 1} \label{algorithmCorrosion_start_beep}
			\State beep 
				\If{$U(0,1)$ $<$ $1 - p$}  \label{algorithmCorrosion_probability_beep}
					\State IsAllowedToBeep $\gets$ 0
				\EndIf
			\EndIf	\label{algorithmCorrosion_end_beep}
			\If{hear at least a beep}
				\State hearBeep($l_k$) $\gets$ 1
			\EndIf
		\EndFor
		\If{$\exists\, m$: hearBeep($l_m$) = 1 \textbf{and} hearBeep($l_n$) = 0 $,\forall n \neq m$}\label{algorithmCorrosion_start_changeValue}  
			\State $v_i(t_{j,K}) \gets l_m$
		\EndIf \label{algorithmCorrosion_end_changeValue}
	\EndFor	
\end{algorithmic}
\end{algorithm}


\begin{definition}
In every phase, we call nodes allowed and not allowed to beep ``alive" nodes and ``dead" nodes, respectively. At the first round of every phase, all nodes are alive. Note that dead nodes can hear but cannot beep until the end of corrosion of that phase.
\end{definition}

If a node is alive, with a chance of $p$, it will remain alive and may beep again in the time slot allocated to its value level in the next round (see line \ref{algorithmCorrosion_probability_beep} in Algorithm \ref{algorithm_corrosion}). Furthermore, if the majority of neighbor nodes of a node have the same value, e.g. $l_m$, we can expect with a good chance that after passing several rounds in a phase, that node will hear beep only at time slot $t_{j,m}$ and will receive no beeps at other time slots of the round. 
Thus, if a node hears a beep at time slot $t_{j,m}$ and does not hear any other beeps at the other time slots of that round, it concludes that the majority of its neighbors may have the same value $l_m$ corresponded to the time slot $t_{j,m}$; thus, it modifies its value to $l_m$ (see lines \ref{algorithmCorrosion_start_changeValue}-\ref{algorithmCorrosion_end_changeValue} in Algorithm \ref{algorithm_corrosion}).

\begin{proposition}
After $\lceil \log_2 (N/\varepsilon) \rceil$ rounds in each phase, all nodes are dead with probability at least $1-\varepsilon$.
\end{proposition}

\begin{proof}
For a node, the probability of being dead after $r$ rounds is $1-p^r$.
Since each node is dying independently, all nodes are dead after $r$ rounds with probability $(1-p^r)^N$.
Moreover, we know $(1-p^r)^N \geq (1-N p^r)$.
Hence, the probability of all nodes being dead after $r$ rounds is at least $(1-N p^r)$ which we want to be greater than $1-\varepsilon$. 
In other words, we need $N p^r \leq \varepsilon$ where $\varepsilon$ is a small constant. Since $p = \frac{1}{2}$, we have $(1/2)^{r} \leq \frac{\varepsilon}{N}$. Thus, all nodes will be dead after $r \geq \log_2 (\frac{N}{\varepsilon})$ rounds with probability at least $1-\varepsilon$. 
\end{proof}
According to the above proposition, we can set the number of rounds in each phase, i.e. $T$, to $\lceil c_1 \log_2 (N) \rceil$ where $c_1$ is a large enough constant, e.g., 20. 

\begin{example}
Consider a network with $N=16$ nodes (see Fig.  \ref{figure_example}). There are three levels of values ($K=3$); thus, every round includes three time slots. 
Figures \ref{figure_example}-c, -d, and -e show possible results of the first, second, and third phases. Several phases are passed to reach consensus in the entire network. In Fig. \ref{figure_example}, the number of nodes having values $l_0$, $l_1$, and $l_2$ are 7, 5, and 4, respectively. Hence, it is expected that the voting algorithm finally converges to value $l_0$. As can be seen in Fig. \ref{figure_example}-f, all nodes reach consensus on value $l_0$ and only one large unique spot has been left.
\end{example}

It is worth mentioning that four cases might happen for a node during a phase: (I) The node remains within its spot; therefore, it will not change its value because its value is equal to the values of all its neighbors (e.g., node 4 in figures \ref{figure_example}-b and \ref{figure_example}-c). (II) The node might be on the boundary of its spot and its spot does not have sufficient merit to be expanded at that point and thus the node changes its value to another value level (e.g., node 2 in figures \ref{figure_example}-b and \ref{figure_example}-c). (III) The node might be on the boundary of its spot and the majority of its neighbors have similar value; therefore, with a good chance, that node will merely hear beep at the time slot of its own value in one of the rounds close to the end of corrosion in that phase. Thus, the node keeps its own value in that phase (e.g., node 7 in figures \ref{figure_example}-b and \ref{figure_example}-c). (IV) The node is on the boundary of its spot but it hears at least two beeps in different time slots of the last round of the phase. In other words, at least its two different neighbors remain alive in the whole phase. In this case, the node does not change its value in that phase (e.g., node 9 in figures \ref{figure_example}-b and \ref{figure_example}-c). 


\subsection{Termination Detection}

In every $D$ phases, the termination detection procedure is executed at the end of the phase to check whether voting is reached (line \ref{algorithmOverall_termination} in Algorithm \ref{algorithm_overall}).
The procedure ``Termination Detection()'' is described in Algorithm  \ref{algorithm_termination}. 
This procedure contains at most $(K-1)$ periods each of which contains $(D + 1)$ time slots.

\begin{definition}
A period $P_k$, allocated to the value level $l_k$, is a sequence of $(D+1)$ time slots, $(t_{k,1}, t_{k,2}, \dots, t_{k,D+1})$.
\end{definition}

\begin{algorithm}[!h]
\caption{Termination Detection}\label{algorithm_termination}
\begin{algorithmic}[1]          
\State Terminated $\gets$ 1
\For{$k$ \textbf{from} $1$ \textbf{to} $K-1$} \label{algorithmTermination_period} 
	\If{$v_i (t_{k,1})$ = $l_k$}  \label{algorithmTermination_start_beep1}
		\State beep     \label{algorithmTermination_end_beep1}
	\ElsIf{hears beep}  \label{algorithmTermination_start_findOut}
	    \State Terminated $\gets$ 0      \Comment Has not terminated   \label{algorithmTermination_end_findOut}
	\EndIf   
	\For{$d$ \textbf{from} $1$ \textbf{to} $D$}  \label{algorithmTermination_diameter}
		\If{Terminated = 0 \textbf{or} hears beep at time $t_{k,d}$}  \label{algorithmTermination_start_terminateCancel}
			\State Terminated $\gets$ 0  \label{algorithmTermination_end_terminateCancel}
			\State beep at time $t_{k,d}$    \Comment Tell others voting has not reached  \label{algorithmTermination_beep_inform}
		\EndIf   
	\EndFor
	\If{Terminated = 0}  
	    \State Break \textbf{procedure} \label{algorithmTermination_break}
	\EndIf
\EndFor
\end{algorithmic}
\end{algorithm}

In the period $P_k$ (line \ref{algorithmTermination_period} in Algorithm \ref{algorithm_termination}), all nodes having value $l_k$ beep at the first time slot of the period (lines \ref{algorithmTermination_start_beep1} and \ref{algorithmTermination_end_beep1} in Algorithm \ref{algorithm_termination}). Meanwhile, if a node hears beep while its value is different from the value of that period, it finds out that there are nodes with different value(s) and consensus has not been reached yet. So, it sets its variable ``Terminated'' to zero (lines \ref{algorithmTermination_start_findOut} and \ref{algorithmTermination_end_findOut} in Algorithm \ref{algorithm_termination}). Every node, which finds out the consensus is not achieved, will beep for the rest of that period (line \ref{algorithmTermination_beep_inform} in Algorithm \ref{algorithm_termination}). Please note that $D$ time slots is sufficient to inform the most distant nodes (line \ref{algorithmTermination_diameter} in Algorithm \ref{algorithm_termination}).

At the end of the period $P_k$, every node checks if the variable ``Terminated'' has been set to zero or not. If it is zero, the algorithm has not terminated yet and needs at least another phase of corrosion. Thus, the node will break the TerminationDetection() (line \ref{algorithmTermination_break} in Algorithm \ref{algorithm_termination}) because it knows that there is no more need for checking termination. Otherwise, it starts a new period with another associated value level. 

Notice that if all  nodes have the same majority value, no one hears a beep in any periods of the procedure. Hence, after $(K-1)$ periods or equally after $(K-1) \times (D + 1)$ time slots, they all decide to terminate. 
Please note that we need $(K-1)$ periods instead of $K$ periods since the $K$-th period is redundant in checking termination.

\section{Distributed Voting with Beeps (DVB) 2 Algorithm}\label{section_DVB2}

The DVB2 algorithm is based on the Distributed Multi-choice Voting/Ranking (DMVR) algorithm \cite{salehkaleybar2015distributed} and results in correct output for any network topology and any number of nodes with high probability (w.h.p.).
In fact, if the nodes already have unique IDs, the DVB2 algorithm eventually returns correct result with probability one. Although, its time complexity might be greater than the one of DVB1 algorithm if the maximum degree of network is large. 
The advantage of DVB2 over DVB1 is that it guarantees returning the correct result while its time complexity is higher than DVB1. Note that DMVR \cite{salehkaleybar2015distributed} assumes that one of the levels is in strict majority; here, the same assumption is considered for DVB2.
It is noteworthy that with a slight change in DVB2, it can be used for distributed ranking of values in beep model.

A schematic of a phase in the DVB2 algorithm is shown in Fig. \ref{figure_diagram2}.
Algorithm \ref{algorithm_overall2} describes the DVB2 algorithm which is performed by every node $i$. 
First, each node chooses a random ID from one to $Y$ (line \ref{algorithmDVB2_chooseID} in Algorithm \ref{algorithm_overall2}). 
\begin{figure*}[!h]
\centering
\includegraphics[width=\textwidth]{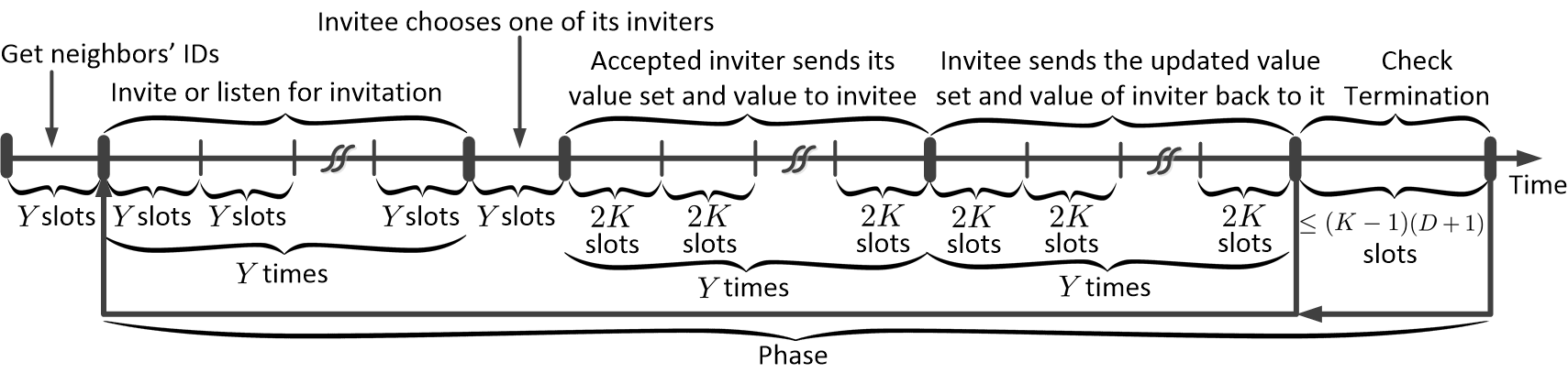}
\caption{A schematic of the phase in DVB2 algorithm. The time-consuming parts of DVB2 are only shown in this diagram. The termination detection is run every $D$ phases, i.e., at the end of every $O(D)$ phases.}
\label{figure_diagram2}
\end{figure*}

According to the URN model \cite{johnson1977urn}, if $Y := \ceil{c_2\, \Delta \log_2 (\Delta)}$ where $\Delta$ is the maximum degree in the network and $c_2$ is a large enough constant (e.g., $20$), then the neighbor nodes will have different IDs w.h.p. It is noteworthy that the probability of having different IDs for neighbors goes to one by increasing $c_2$. This step can be omitted if the nodes already have unique IDs. 

Afterwards, in $Y$ time slots, the nodes find out the IDs of their neighbors (and collecting them in a set $\text{ID}_i^n$) by beeping at the slot associated to their IDs and listening at the other slots (lines \ref{algorithmDVB2_start_neighborID}-\ref{algorithmDVB2_end_neighborID} in Algorithm \ref{algorithm_overall2}). Then, the phases (iterations) start where interactions and termination detection (every $D$ phases) are performed repeatedly until the convergence of voting (lines \ref{algorithmDVB2_start_phase}-\ref{algorithmDVB2_end_phase} in Algorithm \ref{algorithm_overall2}).  
The termination detection in Algorithm \ref{algorithm_overall2} is the same as in the Algorithm \ref{algorithm_termination}. 
Again, in order to improve the time complexity of algorithm, termination detection is executed every $D$ phases.

\begin{algorithm}[!h]
\caption{DVB2 Algorithm}\label{algorithm_overall2}
\begin{algorithmic}[1]
\State /* Choosing a random ID: */ 
\State $\text{ID}_i$ $\sim \text{Uniform}\{1, 2, \dots, Y \}$ \label{algorithmDVB2_chooseID}
\State /* Finding out IDs of neighbors: */ 
\State $\text{ID}_i^n \gets \varnothing$ \label{algorithmDVB2_start_neighborID}
\For{$j$ \textbf{from} 1 \textbf{to} $Y$}  
    \If{$\text{ID}_i = j$}
        \State beep
    \ElsIf{hears beep}
        \State $\text{ID}_i^n \gets \text{ID}_i^n \cup \{j\}$
    \EndIf
\EndFor	\label{algorithmDVB2_end_neighborID}
\State /* DMVR initialization: */ 
\State $\mathcal{V}_i(0) \gets v_i(0)$ \label{algorithmDVB2_valueSetInitialize}
\State /* The phases: */ 
\State counter $\gets$ $0$
\Repeat  \Comment every iteration is a phase \label{algorithmDVB2_start_phase}
    \State counter $\gets$ counter $+ 1$
    \State Interactions($\mathcal{V}_i$)
	\If{counter $= D$}
	    \State counter $\gets$ $0$  \Comment reset the counter
	    \State TerminationDetection() 
	\EndIf
\Until{not terminated}
\label{algorithmDVB2_end_phase}
\end{algorithmic}
\end{algorithm}

The interactions include different steps. First, every node decides to invite or listen for invitation with a probability denoted by $p_\text{inv}$ (lines \ref{algorithmInteractions_start_decideInvite}-\ref{algorithmInteractions_end_decideInvite} in Algorithm \ref{algorithm_interactions1}). 
\begin{definition}
The invitation of an inviter node is defined as a collision because an inviter  cannot be invitee.
\end{definition}
\begin{proposition}\label{proposition_probNotCollision_DVB2}
Suppose that all nodes have the same degree. The probability that an inviter node can interact with one of its neighbor nodes is maximized when $p_\text{inv}$ is equal to $0.5$.
\end{proposition}
\begin{proof}
Each node $i$ becomes inviter with probability $p_\text{inv}$. The probability that one of its neighbors being inviter is also $p_\text{inv}$. Furthermore, the probability of being selected by its inviter neighbor is $1/\Delta$. Thus, the probability that at least one of its neighbors invites it, is $\sum_{j=1}^\Delta (p_\text{inv}/\Delta) \leq p_\text{inv}$.  Therefore, the probability of being inviter for node $i$ while its neighbors do not invite it, is $p_\text{inv}\,(1-p_\text{inv})$ which is maximized when $p_\text{inv}=0.5$.
\end{proof}

\begin{algorithm}[!h]
\caption{Interactions (first part)}\label{algorithm_interactions1}
\begin{algorithmic}[1]
	\State /* decide to be inviter or listener: */
	\If{$U(0,1) < 0.5$} \label{algorithmInteractions_start_decideInvite}
	    \State $\text{Inviter} \gets 1$
	\Else
	    \State $\text{Inviter} \gets 0$
	\EndIf \label{algorithmInteractions_end_decideInvite}
	\State /* invite or listen for invitation: */
	\If{Inviter}
	    \State $\text{ID}_i^\text{invitee} \gets $ choose randomly from $\text{ID}_i^n$ \label{algorithmInteractions_chooseInvitee}
	\EndIf
	\State $\text{ID}_i^\text{inviters} \gets \varnothing$  \Comment{used if node is not inviter}
	\For{$j_1$ \textbf{from} 1 \textbf{to} $Y$} \label{algorithmInteractions_start_invitation}
	    \For{$j_2$ \textbf{from} 1 \textbf{to} $Y$} 
	        \If{Inviter \textbf{and} $\text{ID}_i = j_1$ \textbf{and} $\text{ID}_i^\text{invitee} = j_2$}
    	        \State beep
    	    \EndIf
    	    \If{\textbf{not} Inviter \textbf{and} hears beep}
    	        \State $\text{ID}_i^\text{inviters} \gets \text{ID}_i^\text{inviters} \cup \{j_1\}$
    	    \EndIf
	    \EndFor
	\EndFor	\label{algorithmInteractions_end_invitation}
	\State /* invitee chooses an inviter: */
	\If{\textbf{not} Inviter}
	    \State $\text{ID}_i^\text{inviter} \gets $ choose randomly from $\text{ID}_i^\text{inviters}$ \label{algorithmInteractions_chooseInviter}
	\EndIf
	\For{$j$ \textbf{from} 1 \textbf{to} $Y$} \label{algorithmInteractions_start_invitationResponse} 
	    \If{Inviter \textbf{and} $\text{ID}_i = j$ \textbf{and} hears beep}
	        \State accepted $\gets 1$
        \Else
            \State accepted $\gets 0$
	    \EndIf
	    \If{\textbf{not} Inviter \textbf{and} $\text{ID}_i^\text{inviter} = j$}
	        \State beep
	    \EndIf
	\EndFor \label{algorithmInteractions_end_invitationResponse} 
\end{algorithmic}
\end{algorithm}

\begin{algorithm}[!h]
\caption{Interactions (second part)}\label{algorithm_interactions2}
\begin{algorithmic}[1]
	\State /* accepted inviter sends its value set and value (memory) to invitee: */
	\State $\mathcal{V}_i^\text{inviter}(t) \gets \varnothing$ \Comment{used if node is invitee}
	\For{$j$ \textbf{from} 1 \textbf{to} $Y$} \label{algorithmInteractions2_start_sendValueSet} 
	    \State /* send/receive value set: */
	    \For{$k$ \textbf{from} 1 \textbf{to} $K$} 
    	    \If{Inviter \textbf{and} accepted \textbf{and} $\text{ID}_i = j$}
    	        \If{$k \in \mathcal{V}_i(t)$}
    	            \State beep
    	        \EndIf
    	    \EndIf
	        \If{\textbf{not} Inviter \textbf{and} $\text{ID}_i^\text{inviter} = j$ \textbf{and} hears beep} 
    	        \State $\mathcal{V}_i^\text{inviter}(t) \gets \mathcal{V}_i^\text{inviter}(t) \cup \{k\}$
    	    \EndIf
	    \EndFor  
	    \State /* send/receive value (memory): */
	    \For{$k$ \textbf{from} 1 \textbf{to} $K$} 
    	    \If{Inviter \textbf{and} accepted \textbf{and} $\text{ID}_i = j$}
    	        \If{$k = v_i(t)$}
    	            \State beep
    	        \EndIf
    	    \EndIf
	        \If{\textbf{not} Inviter \textbf{and} $\text{ID}_i^\text{inviter} = j$ \textbf{and} hears beep} 
    	        \State $v_i^\text{inviter}(t) \gets k$
    	    \EndIf
	    \EndFor
	\EndFor \label{algorithmInteractions2_end_sendValueSet}
	\State /* DMVR rules by the invitee: */
	\If{\textbf{not} Inviter \textbf{and} $\text{ID}_i^\text{inviter} \neq \varnothing$ (i.e., was invited)} 
	    \State $\mathcal{V}_i(t^+)$, $\mathcal{V}_i^\text{inviter}(t^+)$, $v_i(t^+)$, $v_i^\text{inviter}(t^+) \gets$ DMVR($\mathcal{V}_i(t)$, $\mathcal{V}_i^\text{inviter}(t)$, $v_i(t)$, $v_i^\text{inviter}(t)$) \label{algorithmInteractions2_applyDMVR}
	 \EndIf
\end{algorithmic}
\end{algorithm}

After deciding whether to be inviter or not, every inviter node randomly selects which neighbor node to invite, hoping that the selected node is not an inviter in that phase (line \ref{algorithmInteractions_chooseInvitee} in Algorithm \ref{algorithm_interactions1}). 
We have a loop of $Y$ times each of which contains $Y$ time slots. The indices of outer and inner loops are the IDs of inviter and listener (not inviter) nodes, respectively (lines \ref{algorithmInteractions_start_invitation}-\ref{algorithmInteractions_end_invitation} in Algorithm \ref{algorithm_interactions1}). The inviter node beeps at the time slot associated to both its ID ($\text{ID}_i$) and its chosen neighbor's ID ($\text{ID}_i^\text{invitee}$). The listener node collects the IDs of nodes which have invited it in a set ($\text{ID}_i^\text{inviters}$). After these nested loops, if $\text{ID}_i^\text{inviters} = \varnothing$ for a listener node, it infers that no one has invited it; otherwise, we call that node an invitee. 

The invitee node chooses one of its inviters randomly whose ID is denoted by $\text{ID}_i^\text{inviter}$ (line \ref{algorithmInteractions_chooseInviter} in Algorithm \ref{algorithm_interactions1}). 
In $Y$ time slots indexing the ID of inviter, the invitee node beeps at the time slot associated to its chosen inviter's ID. If the inviter hears a beep in the time slot of its ID, it finds out that its invitation has been accepted (lines \ref{algorithmInteractions_start_invitationResponse}-\ref{algorithmInteractions_end_invitationResponse} in Algorithm \ref{algorithm_interactions1}).

The DVB2 algorithm uses DMVR algorithm \cite{salehkaleybar2015distributed} which guarantees the correct output using a gossiping mechanism. In the DMVR, each node has a value set $\mathcal{V}_i(t)$ and a memory at the current time $t$. In our work, we can consider the value $v_i(t)$ of a node as its memory in DMVR. The value set of every node is initialized to the initial value of the node (line \ref{algorithmDVB2_valueSetInitialize} in Algorithm \ref{algorithm_overall2}).
We have a loop of $Y$ iterations for indexing the ID of inviter and two back-to-back inner loops of $K$ time slots for encoding the value set and the value (lines \ref{algorithmInteractions2_start_sendValueSet}-\ref{algorithmInteractions2_end_sendValueSet} in Algorithm \ref{algorithm_interactions2}). The inviter sends its value set and value to the invitee by these nested loops, and the invitee collects them ($\mathcal{V}_i^\text{inviter}(t)$ and $v_i^\text{inviter}(t)$).

\begin{algorithm}[!h]
\caption{Interactions (third part)}\label{algorithm_interactions3}
\begin{algorithmic}[1]
	 \State /* invitee sends the updated value set and value (memory) of inviter back to it */
	 \If{Inviter \textbf{and} accepted}
        \State $\mathcal{V}_i(t) \gets \varnothing$
     \EndIf
	 \For{$j$ \textbf{from} 1 \textbf{to} $Y$} 
	    \For{$k$ \textbf{from} 1 \textbf{to} $K$} 
	        \If{Inviter \textbf{and} accepted \textbf{and} $\text{ID}_i=j$ \textbf{and} hears beep}
	            \State $\mathcal{V}_i(t^+) \gets \mathcal{V}_i(t) \cup \{k\}$
	        \EndIf
	        \If{\textbf{not} Inviter \textbf{and} $\text{ID}_i^\text{inviter}=j$}
	            \If{$k \in \mathcal{V}_i^\text{inviter}(t^+)$}
    	            \State beep
    	        \EndIf
	        \EndIf
	    \EndFor
	    \For{$k$ \textbf{from} 1 \textbf{to} $K$} 
	        \If{Inviter \textbf{and} accepted \textbf{and} $\text{ID}_i=j$ \textbf{and} hears beep}
	            \State $v_i(t^+) \gets k$
	        \EndIf
	        \If{\textbf{not} Inviter \textbf{and} $\text{ID}_i^\text{inviter}=j$}
	            \If{$k = v_i^\text{inviter}(t^+)$}
    	            \State beep
    	        \EndIf
	        \EndIf
	    \EndFor
    \EndFor
\end{algorithmic}
\end{algorithm}

When the invitee node receives the value set and the value of inviter, it applies the DMVR rules on them (line \ref{algorithmInteractions2_applyDMVR} in Algorithm \ref{algorithm_interactions2}). The DMVR rules \cite{salehkaleybar2015distributed} are given in Algorithm \ref{algorithm_DMVR}.
After applying the DMVR rules, the invitee sends the updated value set and value of the inviter. Afterwards, the inviter collects them using the similar loops (see Algorithm \ref{algorithm_interactions3}). The updated value set and value of a node are denoted by $\mathcal{V}_i(t^+)$ and $v_i(t^+)$, respectively.

\begin{algorithm}[!h]
\caption{DMVR}\label{algorithm_DMVR}
\begin{algorithmic}[1]
    \State \textbf{Inputs:} $\mathcal{V}_1(t)$, $\mathcal{V}_2(t)$, $v_1(t)$, $v_2(t)$
    \State \textbf{Outputs:} $\mathcal{V}_1(t^+)$, $\mathcal{V}_2(t^+)$, $v_1(t^+)$, $v_2(t^+)$
    \State /* consolidation in DMVR: */
    \If{$|\mathcal{V}_1(t)| \leq |\mathcal{V}_2(t)|$}
        \State $\mathcal{V}_1(t^+) \gets \mathcal{V}_1(t) \cup \mathcal{V}_2(t)$
        \State $\mathcal{V}_2(t^+) \gets \mathcal{V}_1(t) \cap \mathcal{V}_2(t)$
    \Else
        \State $\mathcal{V}_1(t^+) \gets \mathcal{V}_1(t) \cap \mathcal{V}_2(t)$
        \State $\mathcal{V}_2(t^+) \gets \mathcal{V}_1(t) \cup \mathcal{V}_2(t)$
    \EndIf
    \State /* dissemination in DMVR: */
    \If{$|\mathcal{V}_1(t^+)| = 1$}
        \State $v_1(t^+) \gets \text{the member of }\mathcal{V}_1(t^+)$
    \EndIf
    \If{$|\mathcal{V}_2(t^+)| = 1$}
        \State $v_2(t^+) \gets \text{the member of }\mathcal{V}_2(t^+)$
    \EndIf
    \State /* speeding up DMVR voting: */
    \If{$|v_1(t^+)| > 1$ \textbf{and} $|v_2(t^+)| > 1$}
        \If{$U(0,1) < 0.5$}
            \State $v_1(t^+) \gets v_2(t)$
        \Else
            \State $v_2(t^+) \gets v_1(t)$
        \EndIf
    \EndIf
\end{algorithmic}
\end{algorithm}

\section{Correctness of Algorithms}\label{section_correctness}

\subsection{Probability of Success for DVB1}

In this section, we first propose two lower bounds for the probability of success of a phase in DVB1 algorithm when the topology is fully connected. Thereafter, we propose the exact probability of success of a phase, but not in a closed form, for binary voting in fully connected networks using Markov chain. We provide empirical results for the probability of success of DVB1 algorithm for different topologies with different number of nodes and initial distribution of values in Section \ref{section_simulation}. 

\subsubsection{Lower Bound on The Probability of Success}

\begin{proposition}\label{proppsition_lower_bound}
In DVB1 algorithm, a lower bound on the probability of success in fully connected networks is:

\noindent
$\underset{r}{\text{max }} \bigg[1 - \Big[(1-p^r)^{\#l_m(0)} +  \binom{\#l_m(0)}{1} p^r (1-p^r)^{\#l_m(0)-1}\Big]\bigg] \times \\  \bigg[\prod_{k=1, k \neq m}^{K} (1-p^r)^{\#l_k(0)} + \sum_{k=1}^K \Big[ \binom{\#l_k(0)}{1} p^r (1-p^r)^{\#l_k(0)-1} \prod_{k'=1, \, k' \neq m, k}^{K} (1-p^r)^{\#l_{k'}(0)} \Big] \bigg] \\
+ \binom{\#l_m(0)}{1} p^r (1-p^r)^{\#l_m(0)-1}   \prod_{k=1, k \neq m}^{K} (1-p^r)^{\#l_k(0)}$,
\newline
where $\#l_m(0) > \#l_k(0), ~~ \forall k \neq m$.
\end{proposition}

\begin{proof}
 Proof in Appendix \ref{section_appendix_A}.
\end{proof}

In order to get more insight on the behavior of DVB1 algorithm, we consider the probability of the event that at least one node remains alive in level $l_m$ and all nodes having other level $l_k \neq l_m$ are dead at a time slot $r$. This probability of success, as another lower bound, is as follows.

\begin{proposition}\label{proppsition_lower_bound_2}
In DVB1 algorithm, a lower bound on the probability of success in fully connected networks is
\begin{align*}
\Bigg(1-\exp\Bigg(-&\sqrt{\frac{\#l_m(0)}{\#l_{m'}(0)}}\Bigg)\Bigg)
\times\\ 
&\exp\Bigg(-\frac{(K-1)\#l_{m'}(0)}{\sqrt{\#l_m(0) \#l_{m'}(0)}-1}\Bigg),
\end{align*}
where $l_{m}$ and $l_{m'}$ are respectively the value levels with the largest and the second largest initial supporters.
\end{proposition}


\begin{proof}
 Proof in Appendix \ref{section_appendix_B}.
\end{proof}

As it is expected, the lower bound of Proposition \ref{proppsition_lower_bound_2} gets close to one as the ratio $\#l_m(0)/\#l_{m'}(0)$ increases. 
The lower bounds in Propositions \ref{proppsition_lower_bound} and \ref{proppsition_lower_bound_2} are given in Fig. \ref{figure_plot_lower_bound} for executing one phase of binary voting in a fully connected network with 100 nodes. 

\begin{corollary}\label{corollary_DVB1_valueProportion}
For $\#l_m(0)\#l_{m'}(0)\geq K^2$, the probability of success of DVB1 is greater than $1-\varepsilon$, if we have
\begin{equation*}
\frac{\#l_m(0)}{\#l_{m'}(0)} \geq  \frac{1}{4}\Big(\ln (1-\varepsilon) + \sqrt{(\ln (1-\varepsilon))^2 + 4K}\Big)^2.
\end{equation*}
\end{corollary}


\begin{proof}
Proof in Appendix \ref{section_appendix_C}.
\end{proof}

\begin{figure}[!t]
\centering
\includegraphics[width=3.45in]{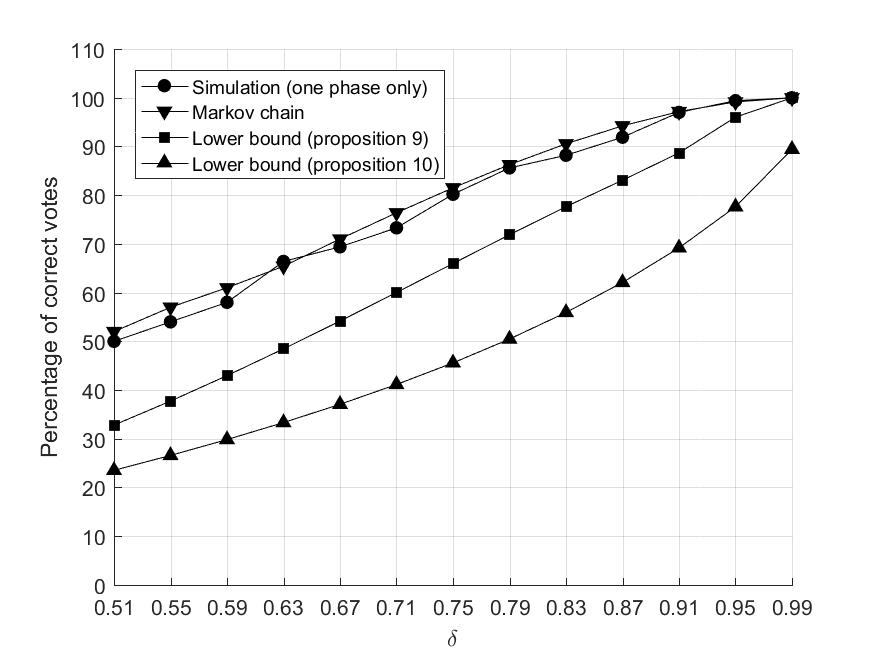}
\caption{The lower bound in Propositions \ref{proppsition_lower_bound} and \ref{proppsition_lower_bound_2}, the simulation result of DVB1, and the result of Markov chain for one phase of binary voting on fully connected topology with 100 nodes.}
\label{figure_plot_lower_bound}
\end{figure}

\subsubsection{Analysis Using Markov Chain}

We study the probability of success for the DVB1 algorithm in fully connected networks using Markov chain model as explained in the following. 
We define $\#l_j$ as the number of alive nodes with level $l_j$, which are still beeping in the current round. We define the state of the algorithm as a vector of remaining alive nodes in each level at the beginning of a round, i.e. $S = (\#l_{1}, \#l_{2}, \dots, \#l_{K})$. The winning states of level $l_m$ are those states where $\#l_m\geq1$, and $\#l_k=0$ for any $k\neq m$ or $\#l_m\geq 2$ and   $\#l_{k_1}\leq 1$ for a $k_1\neq m$ and $\#l_{k_2}=0$ for any $k_2 \neq m$, $k_2 \neq k_1$ (see Appendix \ref{section_appendix_A} for more details). Furthermore, the draw state is equal to the vector of all zeros. We consider all winning states and the draw state as halting states.

Let $\#l_{j1}$ and $\#l_{j2}$ be the number of alive nodes having value $l_j$ in two consecutive rounds (before and after transition between states), respectively. It can be seen that the probability of transition from state $S_1 = (\#l_{11}, \#l_{21}, \dots, \#l_{K1})$ to another state $S_2 = (\#l_{12}, \#l_{22}, \dots, \#l_{K2})$ can be stated as follows (considering $S_1$ is not a halting state),
\begin{equation*}
\prod_{k=1}^K \binom{\#l_{k1}}{\#l_{k1} - \#l_{k2}} (1-p)^{(\#l_{k1} - \#l_{k2})} p^{\#l_{k2}}.
\end{equation*}
An execution of the algorithm can be interpreted as a random walk on this Markov chain which is started from the initial state $(\#l_{1}(0), \#l_{2}(0), \dots, \#l_{K}(0))$ and will be terminated in one of the halting states corresponding to the cases that one of the levels wins or we have a draw. We can compute the probability of success by summing over probabilities of being in the halting states which correspond to the cases that we get the correct result.  

In Figure \ref{figure_plot_lower_bound}, the result of Markov chain is shown for simulating one phase of binary voting on fully connected topology with 100 nodes. As can be seen in this figure, the result of Markov chain for simulating one phase of algorithm almost coincides with the result of simulations (average of multiple runs), as expected.

\subsection{Correctness of DVB2}

\begin{proposition}
The DVB2 algorithm returns the correct majority vote with high probability. 
\end{proposition}
\begin{proof}
According to {\cite[Theorem 1]{salehkaleybar2015distributed}}, the DMVR algorithm returns the correct result if each node can distinguish its neighbors. The DVB2 algorithm implements the DMVR algorithm using beep signals. It assigns IDs to the nodes with a random mechanism which will be unique with high probability, according to the URN model \cite{johnson1977urn}. Therefore, DVB2 algorithm outputs the correct result with high probability. 
\end{proof}

\section{Complexity Analysis of DVB Algorithms}\label{section_complexity}

In this section, we analyze the time, message, and space complexities of the DVB1 and DVB2 algorithms. 
The complexities of the DVB1 and DVB2 algorithms are given in expectation.
We assume that DVB1 algorithm is executed for $\Uptheta(D)$ phases. 
Here, we give an intuition for why such number of phases is needed in the DVB1 algorithm. Consider the largest spot in the network in a given phase. Suppose that  two nodes $i$ and $j$ have the maximum distance among all pairs in the largest spot. We know that both nodes are on the boundary of this spot. Now, consider a node $i'$ which is a neighbor node of $i$ but it is not in the largest spot. Let $S_{majority}$ be the set of neighbor nodes of $i'$ which are in the largest spot and $S_{minority}$ be the set of neighbor nodes of $i'$ that are outside the largest spot.  The probability that node $i'$ joins the largest spot is equal to the probability that at least one node from $S_{majority}$ remains alive until the end of the phase while all nodes from $S_{minority}$ are dead. To bound this probability, we rely on the methodology followed in the previous section. We only need to replace each $\#l_k(0)$ with the number of nodes with initial level $l_k$ around node $i'$.
Assume that the majority of neighboring nodes of $i'$ are in the largest spot. Then, it follows from Fig. \ref{figure_plot_lower_bound} that, with probability $1/2+\epsilon$, for some $\epsilon>0$, node $i'$ will join the largest spot at the end of the phase. Thus, the diameter of the largest spot will be increased by a factor of $\epsilon$ on average. Since the diameter of spot is bounded by the diameter of the network, the largest spot corrodes small spots after $\Uptheta(D)$ phases. 
Experiments in Section 8 also exhibit  this phenomenon.

\subsection{Complexity Analysis of DVB1}


\begin{proposition}\label{proposition_DVB1_time_complexity}
The time complexity of DVB1 algorithm is in the order of $\Uptheta(KD \log (N))$. 
\end{proposition}

\begin{proof}
As can be seen in Fig. \ref{figure_diagram}, the DVB1 algorithm is performed iteratively. In each phase, the corrosion takes $K \times \lceil c_1 \log (N) \rceil$ time slots because it consists of $\lceil c_1 \log (N) \rceil$ rounds, each of which includes $K$ time slots. Since  
we have $\Uptheta(D)$ phases, the complexity of corrosion is $\Uptheta(D K \lceil c_1 \log (N) \rceil) = \Uptheta(KD \log (N))$.
Moreover, at the end of every $D$ phases, up to $(K-1)$ periods are executed, each of which includes $(D + 1)$ time slots. This results in $O(KD)$ and $\Omega(D)$ time slots in the  worst and best cases for checking termination, respectively. Therefore, the time complexity of algorithm is in the order of $\Uptheta(KD \log (N))$. 
\end{proof}


\begin{proposition}\label{proposition_DVB1_spaceComplexity}
The message complexity of DVB1 algorithm is in the order of $O(ND \log (N) + NK)$ and $\Omega(ND)$ bits.
\end{proposition}

\begin{proof}
We consider every beep as one bit message.
The corrosion consists of $\lceil c_1 \log (N) \rceil$ rounds, in each of which every node probably beeps only once when its value is equal to the value of time slot. In the worst case, all nodes beep until the last round of corrosion and in the best case, all nodes die at the first round. Assuming that we have $\Uptheta(D)$ phases, the message complexity of corrosion phase is $O(N D \lceil c_1 \log (N) \rceil)$ and $\Omega(N D)$. In checking termination, there are at most $(K-1)$ periods and at least one period. In each period, in the worst case, each node beeps for zero or one time only. As termination is checked every $D$ phases and we have $O(D)$ phases. Hence, its message complexity is $O(NK)$ and $\Omega(N)$ for $N$ nodes. Therefore, the message complexity of the algorithm is in the order of $O(ND \lceil c_1 \log (N) \rceil + NK) = O(ND \log (N) + NK)$ and $\Omega(ND + N) = \Omega(ND)$.
\end{proof}




\begin{proposition}
The space complexity of DVB1 algorithm is in the order of $O(\log(KD) + \log^{(2)}(N))$ and $\Omega(\log (KD))$, where $\log^{(2)} (.) = \log (\log (.))$.
\end{proposition}

\begin{proof}
Every node should have a memory for counting time slots of corrosion and checking termination. The amount of memory needed for encoding the value is $\Uptheta(\log(K))$, for corrosion is $\Uptheta(\log (K \lceil c_1 \log (N) \rceil))$, for checking termination is $O(\log ((K-1)(D + 1)))$ and $\Omega(\log (D + 1))$, and for counting every $D$ phases for checking termination is $\Uptheta(\log(D))$. Thus, the total required space is $O(\log(K) + \log (K \lceil c_1 \log (N) \rceil) + \log ((K-1)(D + 1)) + \log (D)) = O(\log (K \log (N)) + \log (K D)) = O(\log(KD) + \log^{(2)}(N))$ and $\Omega(\log(K) + \log (K \lceil c_1 \log (N) \rceil) + \log (D + 1) + \log (D)) = \Omega(\log (K) + \log (D)) = \Omega(\log (KD))$. 
\end{proof}

As we will see later in Section \ref{section_simulation_DVB1}, for fully connected network, the DVB1 algorithm does not require checking termination and can be executed solely for one phase. This is because in this topology, all nodes can communicate with one another. Therefore, the complexities of DVB1 algorithm can be relaxed as stated in the following. 

\begin{corollary}
The time, message, and space complexities of DVB1 algorithm for fully connected network are $\Uptheta(K \log (N))$, $O(N \log (N))$, and $\Uptheta\big(\!\log (K \log (N))\big)$, respectively. The lower bound on message complexity for this topology is $\Omega(N)$.
\end{corollary}






\subsection{Complexity Analysis of DVB2}

\begin{proposition}
The time complexity of DVB2 algorithm is in the order of $O\big(D\Delta^2 (\log(\Delta))^2 \log(N) + D K\Delta \log(\Delta) \log(N) \big)$ and $\Omega(\log(N))$.
\end{proposition}
\begin{proof}
As can be seen in Fig. \ref{figure_diagram2}, a phase excluding termination detection includes $\Uptheta(Y^2 + Y + 2YK + 2YK) = \Uptheta(Y^2 + YK)$ time slots. The DMVR algorithm \cite{salehkaleybar2015distributed} has the time complexity $O(D\log(N))$ and $\Omega(\log (N))$ (cf. \cite{shah2009gossip}) for a large class of network topologies when all the nodes randomly interact with one another at every phase. In our work, according to the proof of Proposition \ref{proposition_probNotCollision_DVB2}, the probability that two nodes interact without collision is $p_\text{inv}\,(1-p_\text{inv})=1/4$. Therefore, the number of phases is in the order of $O(D\log(N))$ and $\Omega(\log (N))$ and the total time complexity excluding termination detection is $O(DY^2 \log(N) + DYK\log(N))$ and $\Omega(Y^2 \log(N) + YK\log(N))$.
As the termination detection is performed every $D$ phases, its complexity is $O(K D^2 \log(N) / D) = O(K D \log(N))$ and $\Omega(D \log(N) / D) = \Omega(\log(N))$ (see also the proof of Proposition \ref{proposition_DVB1_time_complexity}).
\end{proof}

\begin{proposition}
The message complexity of DVB2 algorithm is in the order of $O(D KN\log(N))$ and $\Omega(ND + N\log (N))$.
\end{proposition}
\begin{proof}
In finding out neighbors' IDs, we have one beep per node, resulting in $\Uptheta(N)$. In inviting part, $N/2$ nodes are expected to be inviters ($N$ in worst case) each beeping once, resulting in  $\Uptheta(N)$ number of beeps. Similarly, when invitee nodes notify inviters, we have  $\Uptheta(N)$ number of messages. Then, in the worst case, $N$ inviters send $NK$ and $N$  beeps for their value set and value, respectively. 
For the best case, both values would be $N$. 
This also holds for returning the updated value sets and values. 
The number of phases is in the order of $O(D\log(N))$ and $\Omega(\log(N))$.
According to the explanation in proof of Proposition \ref{proposition_DVB1_spaceComplexity}, the complexity of termination detection is $O(NKD^2 \log(N)/D) = O(NKD\log(N))$ and $\Omega(ND)$. Thus, the overall complexity is $O\Big(D\log(N) \big(N + N + N + N(K+1) + N(K+1)\big) + D NK\log(N)\Big) = O(D KN\log(N))$ and $\Omega(ND + N\log (N))$.
\end{proof}

\begin{proposition}
The space complexity of DVB2 algorithm is in the order of $O(K \log(K) + \log(D\Delta))$ and $\Omega(\log(KD) + \log^{(2)} (\Delta))$.
\end{proposition}

\begin{proof}
The amount of memory required in every node for steps shown in Fig. \ref{figure_diagram2} are $\Uptheta(\log(Y))$, $\Uptheta(\log(Y^2))$, $\Uptheta(\log(Y))$, $\Uptheta(\log(KY))$, $\Uptheta(\log(KY))$ and $\Uptheta(\log (KD))$, respectively. Also, counting every $D$ phases for checking termination needs $\Uptheta(\log(D))$ memory. 
Moreover, encoding the value (memory) and the value set require $\Uptheta(\log(K))$ and $O(K \log(K))$ memory, respectively. Encoding the value set requires $\Omega(\log(K))$ in the best case.
Thus, the total required space is $O(\log(Y) + 2\log(Y) + \log(KY) + \log(KD) + \log(D) + \log(K) + K\log(K)) = O(K\log(K) + \log(KDY)) = O(K \log(K) + \log(D\Delta) + \log^{(2)}(\Delta)) = O(K \log(K) + \log(D\Delta))$ and $\Omega(\log(Y) + 2\log(Y) + \log(KY) + \log(KD) + \log(D) + \log(K) + \log(K)) = \Omega(\log(KDY)) = \Omega(\log(KD\Delta) + \log^{(2)} (\Delta)) = \Omega(\log(KD) + \log^{(2)} (\Delta))$.  
\end{proof}

Comparing the complexities of DVB1 and DVB2 shows that DVB1 has better time and message complexities especially when $\Delta$ or $D$ is large.
$K$ is usually a small integer. Hence, if the maximum degree of network is small enough, DVB2 has better space complexity than DVB1; otherwise, DVB1 is better in terms of space complexity.

\subsection{Lower Bound on Time Complexity}

We can have a general lower bound on time of the voting problem in the beep model, considering the constraint of communicating at most one bit per time slot. For completeness, we provide this lower bound below.

\begin{proposition}
The time complexity of distributed voting in the beep model is bounded from below by $\Omega(D + \log (K))$.
\end{proposition}

\begin{proof}

Suppose that in the best-case scenario, one of the nodes, say node $i$, is informed about the majority vote in the first time slot.
Let node $j$ be the farthest node from node $i$ in the network. As nodes can communicate at most one bit per time slot, $\log (K)$ time slots are required for transmitting the encoded message of majority vote. Moreover, at least $\text{dist}(i,j)$ time slots is required to inform node $j$ about the majority vote, where $\text{dist}(i,j)$ denotes the distance between nodes $i$ and $j$. Since $\text{dist}(i,j) \geq D/2$, we can obtain the lower bound $\max(\log (K), D/2) \geq \, (\log (K) + D/2)/2 \geq \, (\log (K) + D)/4= \Omega(D + \log (K))$ on the time complexity of any voting algorithm.

\end{proof}



\section{Comparison to Alternative Approaches}\label{section_comparison}

One can consider some alternative approaches for distributed voting in beep model based on distributed primitives such as leader election or broadcasting. However, these approaches may not be applicable in practice for several reasons such as node failures or stringent resource constraints. For instance, in the literature of gossip algorithms, the main works on computing average \cite{boyd2006randomized,dimakis2010gossip}, voting \cite{benezit2011distributed}, or sum \cite{mosk2008fast,shah2009gossip} do not utilize such distributed primitives. In fact, functions are computed merely by local interactions among nodes and there is a common belief that these solutions are much more robust to single point of failures or unreliable network conditions \cite{dimakis2010gossip}.
However, for the sake of completeness, we compare the complexities of two possible alternative approaches with ours and show that our algorithm performs better in terms of time, space, and message complexities. For fair comparison, we assume upper bounds on $N$ and $D$ are already known to these approaches. In the following, we first briefly describe the distributed primitives that  are used in these alternative approaches and provide their time, message, and space complexities.
Note that the time complexities are based on those reported in the cited papers and the space and message complexities are based on our analysis of those algorithms.

\noindent
\textbf{Leader Election}: In this distributed primitive, exactly one of the nodes is elected as a leader. A time-optimal algorithm for leader election \cite{dufoulon2018brief,dufoulon2018beeping,dufoulon2018beepingDISC} has been recently proposed in beep model.
The eventual version of leader election should be used to infer the time of termination which is essential when it is utilized with other algorithms.
\begin{itemize}
\item Time complexity: $O(D + \log N)$ (cf. {\cite[Theorem 3]{dufoulon2018beeping}}) and $\Omega(\log N)$ (cf. {\cite[Theorem 3.1]{dinitz2007two}} and {\cite[Theorem 1]{casteigts2016deterministic}}).
\item Space complexity: $O(N \log N)$ and $\Omega(\log N)$. 
\item Message complexity: $O(N (D + \log N))$ and $\Omega(N \log N)$.
\end{itemize}


\noindent
\textbf{Broadcasting}: In this distributed primitive, a message with length $\log M$ from a source node is sent to all other nodes in the network. The beep wave algorithm \cite{ghaffari2013near,czumaj2016communicating} has been proposed to send all bits from the source to all other nodes sequentially in time. 
\begin{itemize}
\item Time complexity: $O(D + \log M)$ (cf. {\cite[Lemma 1]{czumaj2016communicating}}) and $\Omega(D)$.
\item Space Complexity: $\Uptheta(\log M)$. 
\item Message complexity: $\Uptheta(N \log M)$. 
\end{itemize}

\noindent
\textbf{Multi-Broadcasting}: In this distributed primitive, a set of source nodes $S$ broadcast their messages, each with length $\log M$ to all nodes in the network. For simplicity in notations, $S$ also denotes the cardinality of this set.
The multi-broadcasting with provenance \cite{czumaj2016communicating,czumaj2019communicating} has been proposed to solve this problem where each node in $S$ sends its message as well as its ID so that nodes can distinguish the origin of each message. In our analysis for this solution, we exclude leader election, diameter estimation, and message length calculation (because $\log M$ (or $\log K$) is known here). 
We also exclude collecting and broadcasting IDs because here all nodes are in $S$ and assuming that nodes are aware of $N$, they can infer the lexicographical order of IDs.


\begin{itemize}
\item Time complexity: $O(D + S \log M)$  and $\Omega(D + S \log M)$ where all nodes are sources (cf. {\cite[Theorems 9 and 11]{czumaj2016communicating}} while excluding the mentioned steps above).  
\item Space complexity: $\Uptheta(\log N + S \log M)$ where all nodes are sources. 
\item Message complexity: $\Uptheta(NS \log M)$.  
\end{itemize}

\noindent
\textbf{Message Gathering}: In this problem, the messages of length $\log M$ from  source nodes in the set $S$ are gathered in the leader node. In order to solve this problem, one can use the multi-broadcasting algorithm in \cite{czumaj2016communicating,czumaj2019communicating} except its last step where the leader broadcasts the messages. Thus, its time and space complexities are equal to the multi-broadcasting algorithm  and its message complexity is $\Uptheta(S^2 \log M)$. Again, in our analysis, we exclude leader election, diameter estimation, message length calculation, and collecting/broadcasting IDs.

Based on the above distributed primitives, we can consider the following two alternative approaches:

\subsection{First Alternative Approach}

The first approach consists of the following steps: (i) One of the nodes is elected as a leader. (ii) All nodes send their values to the leader by executing the message gathering algorithm. (iii) Leader counts the votes and broadcasts the majority vote via the broadcasting algorithm.

By considering $M=K$ and $S=N$, the time, space, and message complexities of this approach are $O(D + \log N + N \log K)$, $O(N \log (KN))$, and $O(ND + N \log N + N^2 \log K)$, respectively. 
The lower bounds on the time, space, and message complexities are $\Omega(D + \log N + N \log K)$, $\Omega(\log N + N \log K)$, and $\Omega(N \log N + N^2 \log K)$, respectively.

\subsection{Second Alternative Approach}

The second approach includes these steps: (i) One of the nodes is elected as a leader. (ii) All nodes broadcast their values to all nodes in the network via multi-broadcasting. (iii) Each node counts the votes and obtains the majority vote itself.
For this approach, the upper bounds and lower bounds on  the time, space, and message complexities are equal to those of the first approach.




\subsection{Comparison}

Comparing the time, space, and message complexities of the two alternative approaches with the ones for DVB1 algorithm (which are $\Uptheta(KD\log(N))$ for time, $O(\log(KD) + \log^{(2)}(N))$ and $\Omega(\log KD)$ for space, and $O(ND\log(N) + NK)$ and $\Omega(ND)$ for message) shows that the proposed algorithm performs better in terms of time, space, and message complexities. 
Note that $K $ is most often a small integer and we usually have $D \ll N$.

If we compare the two alternative approaches to the time, space, and message complexities of DVB2 algorithm (which are $O\big(D\Delta^2 (\log(\Delta))^2 \log(N) + KD\Delta \log(\Delta) \log(N) \big)$ and $\Omega(\log N)$ for time, $O(K \log(K) + \log(D\Delta))$ and $\Omega(\log(KD) + \log^{(2)} \Delta)$ for space, and $O(KDN\log(N))$ and $\Omega(ND + N \log N)$ for message), we see that DVB2 outperforms in space and message complexities. It is also better in time complexity if $\Delta$ is small (e.g., in mesh grid or torus networks). Moreover, DVB2 is not prone to fail because of possible failure of leader and also does not require pre-defined unique IDs.


\begin{figure}[!t]
\centering
\includegraphics[width=3.15in]{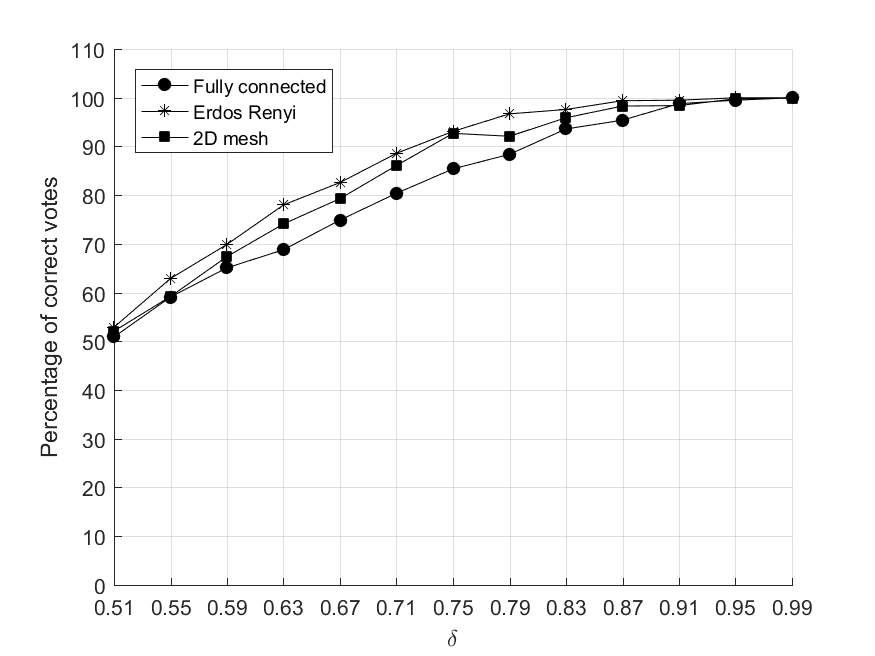}
\caption{Simulation of DVB1 algorithm on fully connected, 2D mesh grid, and Erd{\"o}s-R{\'e}nyi networks with 100 nodes for binary voting.}
\label{figure_simulation_twoValues}
\end{figure}

\begin{figure}[!t]
\centering
\includegraphics[width=3.15in]{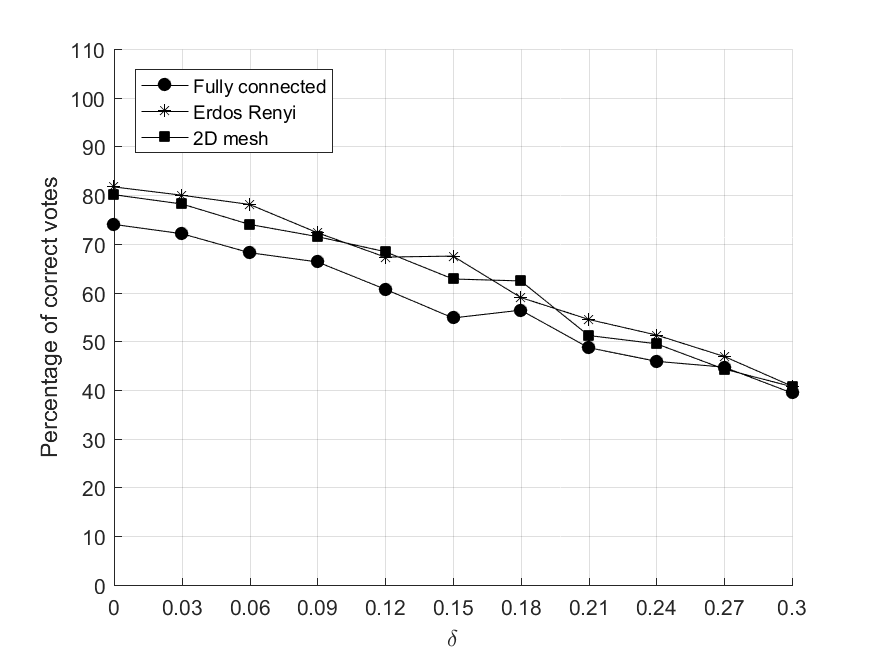}
\caption{Simulation of DVB1 algorithm on fully connected, 2D mesh grid, and Erd{\"o}s-R{\'e}nyi networks with 100 nodes for ternary voting.}
\label{figure_simulation_threeValues}
\end{figure}

\section{Simulations}\label{section_simulation}

\subsection{Simulations for DVB1}\label{section_simulation_DVB1}

\subsubsection{Experiments on the size of majority level}\label{section_simulation_on_populationMajority}

We simulated DVB1 algorithm on fully connected, 2D mesh grid, and Erd{\"o}s-R{\'e}nyi (with probability $\frac{2}{N} \log_2(N)$ of having an edge between any pair of nodes) networks with $100$ nodes. Figures \ref{figure_simulation_twoValues} and \ref{figure_simulation_threeValues} show the average results of simulation over 1000 number of experiments for binary and ternary value levels, respectively. In experiments, the initial values are assigned randomly to the nodes such that the percentage of initial values for different levels are $(1 - \delta, \delta)$ and $(\frac{2}{3} - \delta, \frac{1}{3}, \delta)$ for binary and ternary voting, respectively. The range of $\delta$ is $ (\frac{1}{2},1]$ for binary and $[0, \frac{1}{3})$ for ternary voting. The $D$ is considered to be its upper bound $N$ in the simulations.
As can be seen in Figures \ref{figure_simulation_twoValues} and \ref{figure_simulation_threeValues}, the proposed algorithm has an acceptable robustness with respect to both topology of network and distribution of majority value in the network. 

\begin{figure}[!t]
\centering
\begin{subfigure}[b]{0.49\textwidth}
\centering
\includegraphics[width=3.15in]{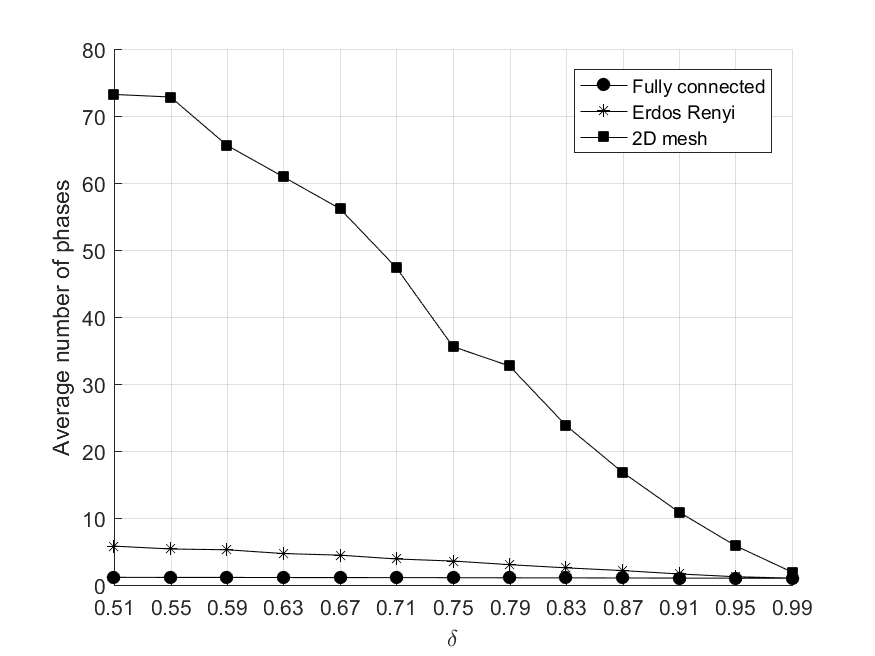} 
\caption{Binary voting}
\label{figure_simulation_twoValues_phases_DVB1}
\end{subfigure}
\begin{subfigure}[b]{0.49\textwidth}
\centering
\includegraphics[width=3.15in]{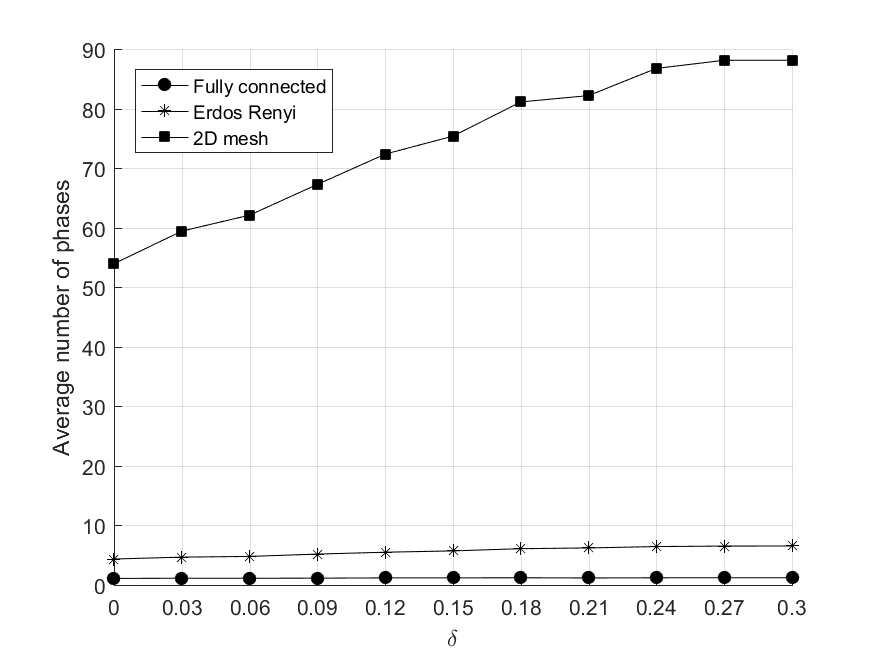} 
\caption{Ternary voting}
\label{figure_simulation_threeValues_phases_DVB1}
\end{subfigure}
\caption{Average number of consumed phases in simulations of DVB1 algorithm on fully connected, 2D mesh grid, and Erd{\"o}s-R{\'e}nyi networks with 100 nodes for binary and ternary voting.}
\end{figure}

The average number of phases until reaching consensus on the majority vote in the simulations of Figures \ref{figure_simulation_twoValues} and \ref{figure_simulation_threeValues} are depicted in Figures \ref{figure_simulation_twoValues_phases_DVB1} and \ref{figure_simulation_threeValues_phases_DVB1}, respectively. As can be seen in Fig. \ref{figure_simulation_twoValues_phases_DVB1}, the average time decreases by increasing $\delta$ because when $\delta$ is close to 0.5, the population of majority value is close to the other value and it takes more time for DVB1 algorithm to reach consensus on the majority vote. The same analysis holds for $\delta$ close to $1/3$ in Fig. \ref{figure_simulation_threeValues_phases_DVB1}.
Moreover, as shown in Figures \ref{figure_simulation_twoValues_phases_DVB1} and \ref{figure_simulation_threeValues_phases_DVB1}, the average number of phases until reaching consensus on the majority vote is almost always one in fully connected topology. We expect with an acceptable chance that voting will be achieved after only one phase since we give the algorithm sufficient number of rounds for convergence, as mentioned before. Therefore, for fully connected networks, the DVB1 algorithm might be performed with only one phase excluding termination detection.



\subsubsection{Experiments on the number of nodes}

We performed another experiment by changing the number of nodes in the network between 20-100 nodes for binary voting with $\delta=2/3$. Figure \ref{figure_simulation_numberOfNodes} shows the average results of 1000 number of simulations for this experiment. The $D$ was set to its upper bound, i.e. $N$, in the simulations. As can be seen in this figure, the DVB1 algorithm is almost invariant with respect to the number of nodes in the network. Notice that this property also exists for ternary voting but for the sake of brevity, the binary case is just shown here. 

Overall, it was empirically shown that the DVB1 algorithm is fairly invariant to the following network attributes:
\begin{itemize}
\item The topology of network
\item The initial distribution of values in the network
\item The number of nodes
\end{itemize}

\begin{figure}[!t]
\centering
\includegraphics[width=3.25in]{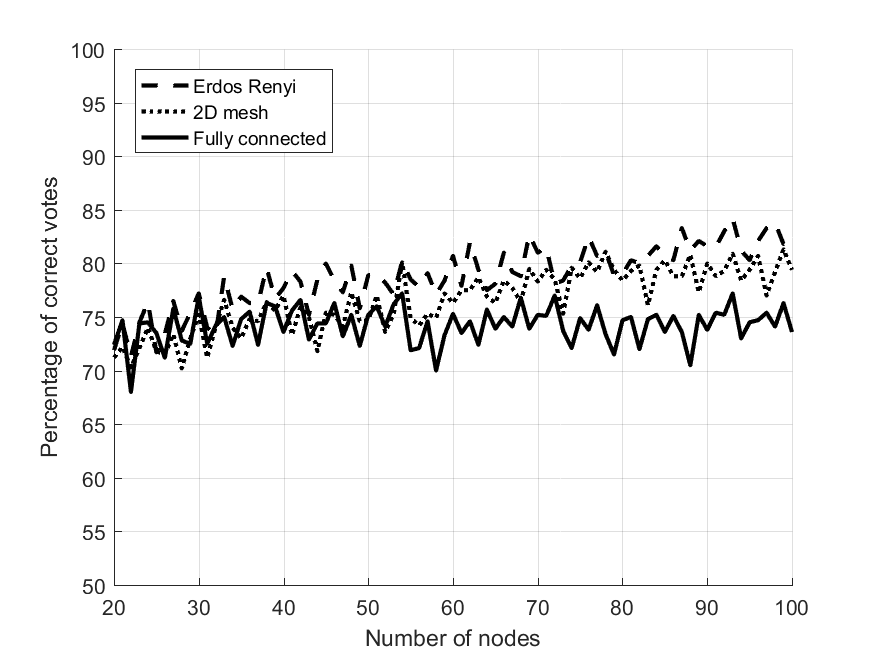}
\caption{Simulation of DVB1 algorithm on fully connected, 2D mesh, and Erd{\"o}s-R{\'e}nyi networks with different number of nodes for binary voting.}
\label{figure_simulation_numberOfNodes}
\end{figure}

In addition to the above experiments, we also depict the number of phases versus the network size for the three network topologies in Figure \ref{figure_simulation_PhasesVsNodes}. As can be seen, the number of phases is in the order of $\Uptheta(1)$, $\Uptheta( \sqrt{N})$, and $\Uptheta(\log(N))$ (where $c$ is a constant) for the fully connected, 2D mesh grid, and Erd{\"o}s-R{\'e}nyi networks, respectively. This validates the assumption of $\Uptheta(D)$ number of phases in the complexity analysis of DVB1. Moreover, the assumption of having only one phase of DVB1 algorithm in fully connected topology is also verified in this figure.

\begin{figure}[t]
	\centering
	\begin{subfigure}{0.49\textwidth}
		\includegraphics[width=2.9in]{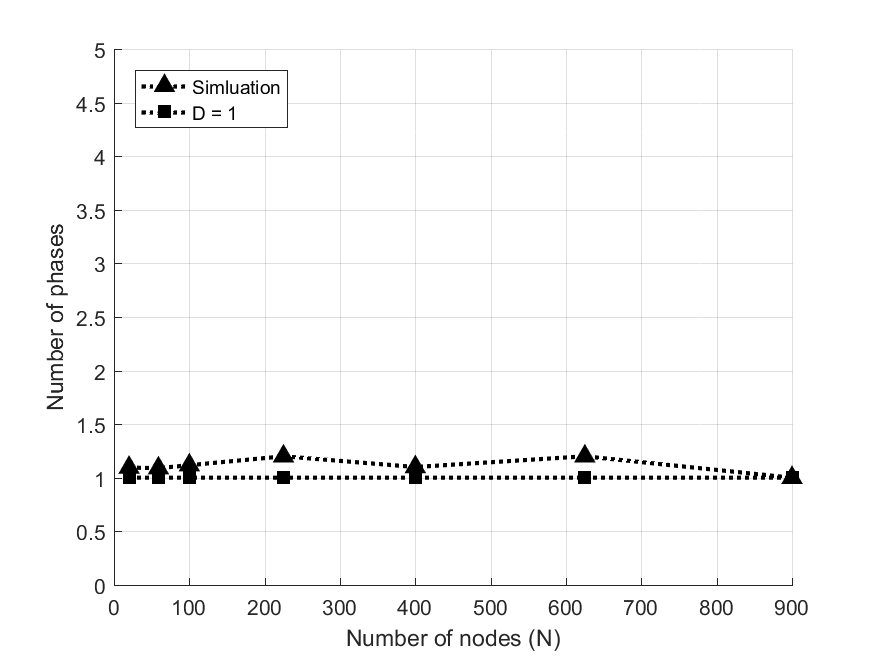}
		\caption{}
	\end{subfigure}
	\begin{subfigure}{0.49\textwidth}
		\includegraphics[width=2.9in]{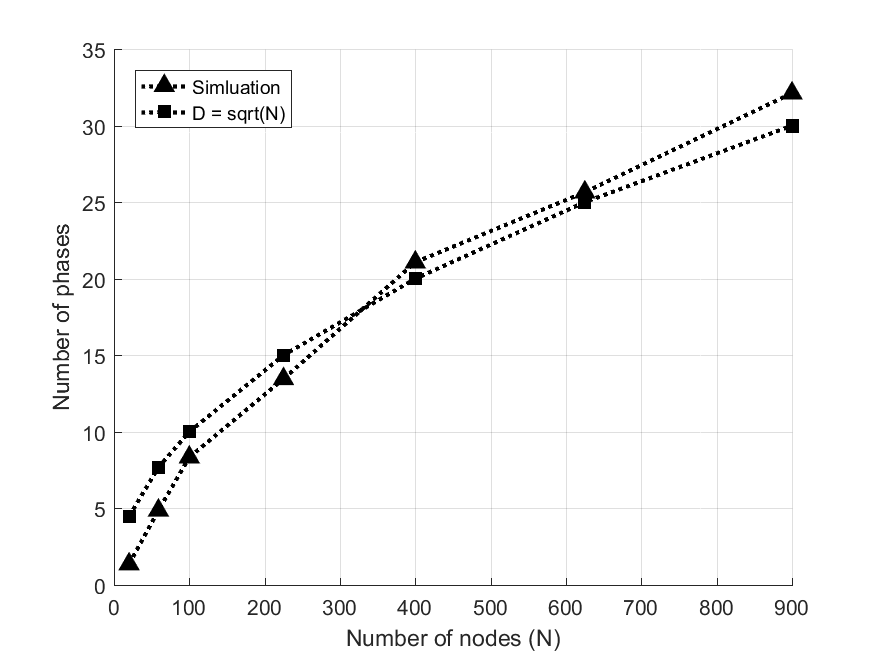}
		\caption{}
	\end{subfigure}
	\begin{subfigure}{0.49\textwidth}
		\includegraphics[width=2.9in]{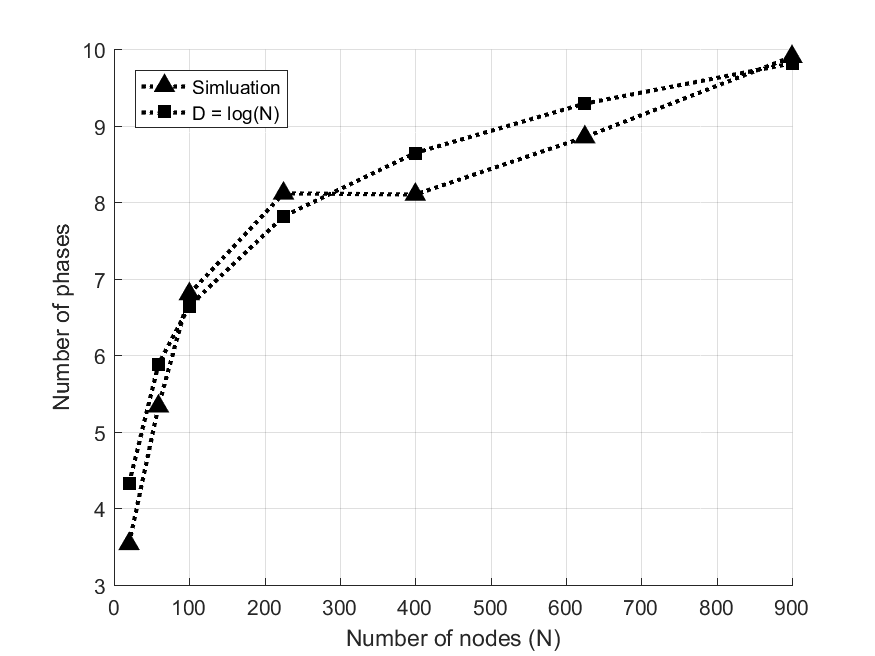}
		\caption{}
	\end{subfigure}
	\caption{The number of phases of DVB1 algorithm in (a) fully connected, (b) 2D mesh grid, and (c) Erd{\"o}s-R{\'e}nyi networks with different number of nodes for the task of binary voting. The number of phases are scaled by a constant for better illustration of the order of complexity.}
	\label{figure_simulation_PhasesVsNodes}
\end{figure}

\subsection{Simulations for DVB2}

\begin{figure}[!t]
\centering
\begin{subfigure}[b]{0.49\textwidth}
\centering
\includegraphics[width=3in]{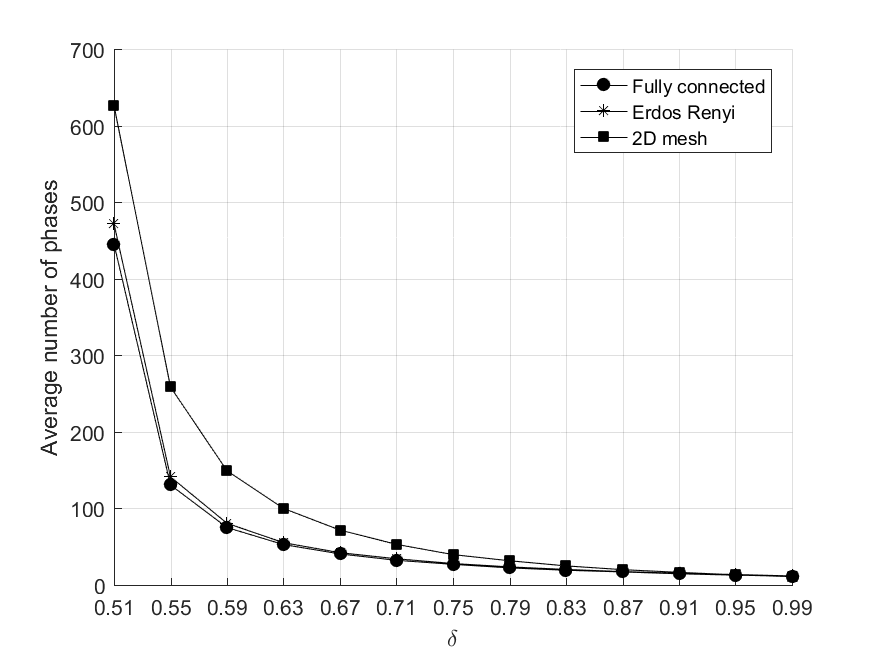} 
\caption{Binary voting}
\label{figure_simulation_twoValues_phases_DVB2}
\end{subfigure}
\begin{subfigure}[b]{0.49\textwidth}
\centering
\includegraphics[width=3in]{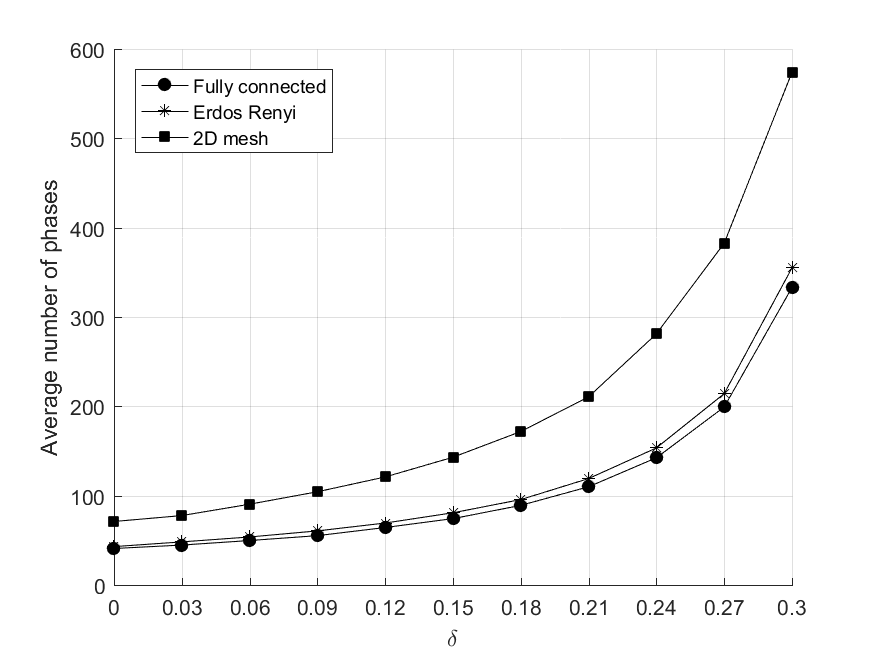} 
\caption{Ternary voting}
\label{figure_simulation_threeValues_phases_DVB2}
\end{subfigure}
\caption{Average number of required phases in simulations of DVB2 algorithm on fully connected, 2D mesh grid, and Erd{\"o}s-R{\'e}nyi networks with 100 nodes for binary and ternary voting.}
\end{figure}

We simulated the DVB2 algorithm on fully connected, 2D mesh grid, and Erd{\"o}s-R{\'e}nyi with probability $\frac{2}{N} \log_2(N)$. The number of nodes and the number of simulations for any value of $\delta$ were 100 and 1000, respectively. The experiments and the sweeps over $\delta$ are similar to the setting in Section \ref{section_simulation_on_populationMajority}. According to \cite{salehkaleybar2015distributed}, the DVB2 is invariant to the topology, the initial distribution of values, and the network size, and it returns the correct output with probability one.
Figures \ref{figure_simulation_twoValues_phases_DVB2} and \ref{figure_simulation_threeValues_phases_DVB2} show the average number of phases until reaching consensus in the simulations. As expected, the average number of phases increases when $\delta$ is close to $0.5$ (for binary voting) and $1/3$ (for ternary voting). This is because in these cases, the number of value levels are very close to each other requiring more time for reaching consensus. Moreover, comparing Figures \ref{figure_simulation_twoValues_phases_DVB1} and \ref{figure_simulation_threeValues_phases_DVB1} with \ref{figure_simulation_twoValues_phases_DVB2} and \ref{figure_simulation_threeValues_phases_DVB2} shows that the DVB2 algorithm requires more phases than DVB1 to terminate which makes sense because of its time complexity. On the other hand, the DVB2 algorithm returns the correct output regardless of initial distribution of values or network topology.

\section{Conclusion \& Future Work}\label{section_conclusion}
In this paper, we proposed two simple-structure algorithms for distributed voting in beep model.  The first algorithm is based on forming spots of nodes having the same values, and corroding the small (non-majority) spots against the large (majority) spots to reach the correct result.  Our experiments showed that this algorithm is fairly invariant  to the network topology, initial distribution of values, and the size of network.
The second algorithm is based on the DMVR algorithm \cite{salehkaleybar2015distributed} and is also invariant to topology, initial values, and the network size.
As a possible future work, one can study the problem of distributed voting using beep signals in the asynchronous model. 



%

\appendices


\section{Proof of Proposition \ref{proppsition_lower_bound}}\label{section_appendix_A}

In fully connected networks, the DVB1 algorithm returns correct result in both of the following events:  
\begin{itemize}
\item At least two nodes of level $l_m$ remain alive while there is at most one node alive in one of the other values $k \neq m$, or
\item One alive node remains in level $l_m$ and all nodes in other values $l_k \neq l_m$ are dead,
\end{itemize}
where $\#l_m(0) > \#l_k(0)$, $\forall k \neq m$.

In the first event, the value $l_m$ wins definitely because the only alive node in another value $l_k$ ($k \neq m$) hears beep from the alive nodes in value $l_m$ (and not hears itself) in the current round and it changes its value to $l_m$. However, the case in which there exists one alive node in each level should be deferred waiting to next rounds until all except one of them die (so that level wins) or all die. In the second event, level $l_m$ wins because in the current round, all the nodes will hear beep from the alive node with value $l_m$ and thus change their values to $l_m$.

The probability of having at least two nodes alive in level $l_m$ at round $r$ (i.e., not having all dead or all except one dead) is $1 - \Big[(1-p^r)^{\#l_m(0)} +  \binom{\#l_m(0)}{1} p^r (1-p^r)^{\#l_m(0)-1}\Big]$.
The probability of having all nodes dead in other levels $l_k$ ($k \neq m$) at round $r$ is $\prod_{k=1, k \neq m}^{K} (1-p^r)^{\#l_k(0)}$ as the probabilities are independent.
Moreover, at round $r$, The probability of having one alive node in value level $l_k$ ($k \neq m$)  while all nodes of other value levels are dead is
$\binom{\#l_k(0)}{1} p^r (1-p^r)^{\#l_k(0)-1} \prod_{k'=1, \, k' \neq m, k}^{K} (1-p^r)^{\#l_{k'}(0)}$. This expression should be summed up over all value levels to cover all possibilities. 
Finally, the probability of the first event is 
$\bigg[1 - \Big[(1-p^r)^{\#l_m(0)} +  \binom{\#l_m(0)}{1} p^r (1-p^r)^{\#l_m(0)-1}\Big]\bigg] \times  \bigg[\prod_{k=1, k \neq m}^{K} (1-p^r)^{\#l_k(0)} + \sum_{k=1}^K \Big[ \binom{\#l_k(0)}{1} p^r (1-p^r)^{\#l_k(0)-1} \prod_{k'=1, \, k' \neq m, k}^{K} (1-p^r)^{\#l_{k'}(0)} \Big] \bigg]$.

The probability of the second event is  $\binom{\#l_m(0)}{1} p^r (1-p^r)^{\#l_m(0)-1}  \prod_{k=1, k \neq m}^{K} (1-p^r)^{\#l_k(0)}$.
The overall probability of success at round $r$ is at least $\bigg[1 - \Big[(1-p^r)^{\#l_m(0)} +  \binom{\#l_m(0)}{1} p^r (1-p^r)^{\#l_m(0)-1}\Big]\bigg] \times  \bigg[\prod_{k=1, k \neq m}^{K} (1-p^r)^{\#l_k(0)} + \sum_{k=1}^K \Big[ \binom{\#l_k(0)}{1} p^r (1-p^r)^{\#l_k(0)-1} \prod_{k'=1, \, k' \neq m, k}^{K} (1-p^r)^{\#l_{k'}(0)} \Big] \bigg] + \binom{\#l_m(0)}{1} p^r (1-p^r)^{\#l_m(0)-1}  \prod_{k=1, k \neq m}^{K} (1-p^r)^{\#l_k(0)}$.
However, note that this lower bound is correct for any $r \geq 0$. So, we can get a better lower bound on the probability of success by maximizing over $r$ and obtain the bound stated in the proposition.

\section{Proof of Proposition \ref{proppsition_lower_bound_2}}\label{section_appendix_B}

If in a time slot, at least one node remains alive in level $l_m$ (i.e., not having all of them dead) and all nodes having other level $l_k \neq l_m$ are dead,  voting is done correctly.
By maximizing over $r$, the probability of this event at a time slot $r$ is obtained as
\begin{align*}
\max_r \Big(1-(1-p^r)^{\#l_m(0)}\Big) \times \prod_{l_k\neq l_m} (1-p^r)^{\#l_k(0)}.
\end{align*}
The above expression is a lower bound on the probability of success. 
Consider $r=\lceil \log_2(\sqrt{\#l_m(0) \#l_{m'}(0)})\rceil$. Assuming $p = 1/2$ and substituting the value of $r=\lceil \log_2(\sqrt{\#l_m(0) \#l_{m'}(0)})\rceil$ in the term $(1 - p^r)$, we have $(1 - p^r) = 1 - 1/\sqrt{\#l_m(0) \#l_{m'}(0)}$. 
We know that $(1 - p^r)^{\#l_m(0)} = \exp(\#l_m(0) \ln(1 - p^r)) = \exp(\#l_m(0) \ln(1 - 1/\sqrt{\#l_m(0) \#l_{m'}(0)}))$. 
Because of $\ln(1 + x) \leq x$ for $x > -1$, we have $(1 - p^r)^{\#l_m(0)} \leq \exp(-\#l_m(0) / \sqrt{\#l_m(0) \#l_{m'}(0)}) = \exp(-\sqrt{\#l_m(0) / \#l_{m'}(0)})$. 
So, the first part of expression, which was $\Big(1-(1-p^r)^{\#l_m(0)}\Big)$, is greater than or equal to $\Big(1 - \exp(-\sqrt{\#l_m(0) / \#l_{m'}(0)})\Big)$. 

Now, we consider the second part of expression which is $\prod_{l_k\neq l_m} (1-p^r)^{\#l_k(0)}$. 
Similar to the previous approach, we have $(1 - p^r)^{\#l_k(0)} = \exp(\#l_k(0) \ln(1 - 1/\sqrt{\#l_m(0) \#l_{m'}(0)}))$. 
As we have $x / (1 + x) \leq \ln(1 + x)$ for $x > -1$, we can say $(1 - p^r)^{\#l_k(0)} \geq \exp(-\#l_k(0) / (\sqrt{\#l_m(0) \#l_{m'}(0)} - 1) )$.
We rewrite the second expression as $\max_r \prod_{l_k\neq l_m} (1-p^r)^{\#l_k(0)} \geq \exp(- \sum_{l_k\neq l_m} \#l_k(0) / (\sqrt{\#l_m(0) \#l_{m'}(0)} - 1) ) \geq \exp(- (K-1) \#l_{m'}(0) / (\sqrt{\#l_m(0) \#l_{m'}(0)} - 1) )$.

Finally, substituting the two derived expressions in the first and second expressions of $\Big(1-(1-p^r)^{\#l_m(0)}\Big) \times \prod_{l_k\neq l_m} (1-p^r)^{\#l_k(0)}$ results in the proposed lower bound.

\section{Proof of Corollary \ref{corollary_DVB1_valueProportion}}\label{section_appendix_C}

Assuming that $\#l_m(0)\#l_{m'}(0)\geq K^2$, the probability of success, proposed in Proposition \ref{proppsition_lower_bound_2}, becomes greater than $(1 - \exp(-y)) \times \exp(-K/y)$ where $y = \sqrt{\#l_m(0) / \#l_{m'}(0)}$. 
Thus, we need to have $(1 - \exp(-y)) \times \exp(-K/y) \geq 1-\varepsilon$. Taking logarithm from both sides of this inequality, we have $\ln(1 - \exp(-y)) - K/y \geq \ln (1-\varepsilon)$. 
As we have $y > 0$ and $0 < \exp(-y) < 1$, we can say $\ln(1 - \exp(-y)) \leq -\exp(-y) \leq y$. Therefore, $y - K/y \geq \ln(1 - \exp(-y)) - K/y \geq \ln (1-\varepsilon)$ which results in $y \geq 0.5\Big(\ln (1-\varepsilon) + \sqrt{(\ln (1-\varepsilon))^2 + 4K}\Big)$.





\ifCLASSOPTIONcaptionsoff
  \newpage
\fi



%

\bibliographystyle{IEEEtran}
\bibliography{References}

%

\hfill \break

\begin{IEEEbiography}[{\includegraphics[width=1in,height=1.25in,clip,keepaspectratio]{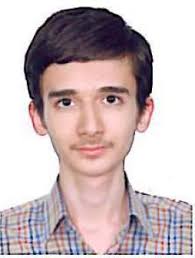}}]%
{Benyamin Ghojogh}
received the first and second B.Sc. degrees in electrical engineering (electronics and telecommunications) from the Amirkabir University of Technology, Tehran, Iran, in 2015 and 2017, respectively, and the M.Sc. degree in electrical engineering from the Sharif University of Technology, Tehran, Iran, in 2017. He is currently pursuing the Ph.D. degree in electrical and computer engineering with the University of Waterloo, Waterloo, ON, Canada. His research interests include distributed algorithms, machine learning, and manifold learning.
\end{IEEEbiography}

\begin{IEEEbiography}[{\includegraphics[width=1in,height=1.25in,clip,keepaspectratio]{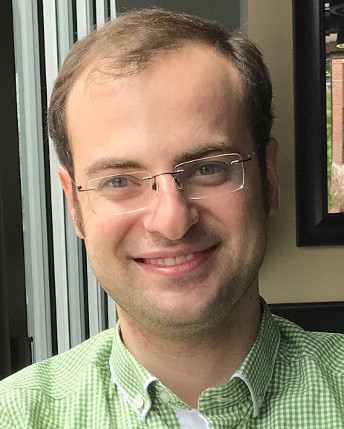}}]%
{Saber Salehkaleybar}
received the B.Sc., M.Sc., and Ph.D. degrees in Electrical Engineering from Sharif University of Technology, Tehran, Iran, in 2009, 2011, and 2015, respectively. He is currently an Assistant Professor of Electrical Engineering at Sharif University of Technology, Tehran, Iran. Prior to joining Sharif University of Technology, he spent a year as a postdoctoral researcher working on causal inference in Coordinated Science Lab. (CSL) at University of Illinois at Urbana-Champaign (UIUC). His research interests include distributed systems, machine learning, and causal inference.
\end{IEEEbiography}

\end{document}